\documentclass[10pt,journal]{IEEEtran}

\ifCLASSINFOpdf

\else

\fi

\usepackage{amsmath}
\usepackage{amsthm}

\newtheorem{Lemma}{Lemma}
\newtheorem{Theorem}{Theorem}
\newtheorem*{Proof}{Proof}

\usepackage{svg}
\usepackage{subfigure}
\usepackage{amssymb}
\usepackage{booktabs}
\usepackage{multirow}
\begin{document}

\title{Evolving Network Modeling Driven by the Degree Increase and Decrease Mechanism}

\author{Yuhan Li, Minyu Feng {\em Member, IEEE}, J\"{u}rgen Kurths

\thanks{Manuscript received 27 September 2022; revised 28 January 2023; accepted 12 April 2023. Date of publication 1 May 2023; date of current version 18 August 2023. This work was supported in part by the National Nature Science Foundation of China (NSFC) under Grant No. 62206230,
and in part by the Humanities and Social Science Fund of Ministry
of Education of the People's Republic of China under Grant No. 21YJCZH028. This article
was recommended by Associate Editor H. Zhu. \textit{(Corresponding author:
Minyu Feng.)} }

\thanks{Yuhan Li, Minyu Feng (corresponding author) are with College
of Artificial Intelligence, Southwest University, Chongqing 400715, P. R. China.
(e-mail: myfeng@swu.edu.cn).
}

\thanks{J\"{u}rgen Kurths is with the Potsdam Institute for Climate Impact Research,
14437 Potsdam, Germany.}

\thanks{Color versions of one or more figures in this article are available at https://doi.org/10.1109/TSMC.2023.3268372.}

\thanks{Digital Object Identifier 10.1109/TSMC.2023.3268372}}

\mark{IEEE TRANSACTIONS ON SYSTEMS, MAN, AND CYBERNETICS: SYSTEMS, VOL. 53, NO. 9, SEPTEMBER 2023 }%

\maketitle

\begin{abstract}
Ever since the Barab\'{a}si-Albert (BA) scale-free network has been proposed, network modeling has been studied intensively in light of the network growth and the preferential attachment (PA). However, numerous real systems are featured with a dynamic evolution including network reduction in addition to network growth. In this paper, we propose a novel mechanism for evolving networks from the perspective of vertex degree. We construct a queueing system to describe the increase and decrease of vertex degree, which drives the network evolution. In our mechanism, the degree increase rate is regarded as a function positively correlated to the degree of a vertex, ensuring the preferential attachment in a new way. Degree distributions are investigated under two expressions of the degree increase rate, one of which manifests a ``long tail'', and another one varies with different values of parameters. In simulations, we compare our theoretical distributions with simulation results and also apply them to real networks, which presents the validity and applicability of our model.
\end{abstract}


\IEEEpeerreviewmaketitle

\section{INTRODUCTION}
A great number of real-life systems have a structure composed of links and elements changing over time, e.g. the Internet, biological organisms, and populations. Complex networks have been playing a significant role in describing real complex systems, revealing the essential properties of these systems. As more and more dynamic networks emerge, modeling of evolving networks and analyses for their topological properties are crucial, which lay the basis for studying network dynamics.

Pioneering studies have given important insights into evolving network models and analyses for topological properties. There are some dynamic network models, e.g., temporal networks where vertices are not constantly connected, adaptive networks where vertices break their connections to neighbors autonomously according to epidemic spreading, and some evolving networks considering node addition and deletion. However, evolving network modeling is still worth studying to construct a more realistic network and also analyze its critical properties. Temporal network models and adaptive network models are not available to describe those networks including the growth and decrease of vertices, and current evolving network models which consider the growth and decrease of vertices are not in favor of investigating the important topological property, the degree distribution. To address these issues, in this paper, we construct evolving network driven by the variance of degree, endowed with dynamics and randomness based on stochastic process and queueing theory. The realization of preferential attachment is not expressed as a probability in the BA model, instead, we propose a degree increase and decrease mechanism, where the increase rate is proportional to the present degree. The main contributions of this paper are as follows.

\begin{enumerate}
\item We establish a novel evolving network model with node growth and deletion driven by the increase and decrease of degrees. Our proposed network model provides a theoretical framework for describing self-evolutionary behaviors of complex systems featured with network structures in real life. Based on this model, we can analyze real-world networks' properties and better understand their evolution patterns.  
\item Our proposed network model provides a new perspective on calculating degree distributions of evolving networks. As a key topological property, degree distributions have been difficult to obtain in a form of analytical solutions, in particular, for those evolving networks with the growth and deletion of nodes. We adjust the traditional PA mechanism by regarding that the degree increase rate of a node is proportional to its degree. To this end, it is easier to obtain specific expressions for degree distributions of evolving networks.
\item Based on our model, we produce two new probability distributions that can describe part of real-world evolving networks in different fields well. There are a great majority of networks in real life whose degree distributions do not follow Power-law distributions, while our theoretically obtained distributions can better fit these real-world networks compared to Power-law distributions according to our simulations.
\end{enumerate}


The organization of the paper is as follows: In Section \ref{sec:I}, we review related works about evolving network models and methods of analyzing topological properties. In Section \ref{sec:II}, we display our evolving network model and perform analyses for degree distributions of the network. In Section \ref{sec:III}, simulations are carried out to demonstrate the validity of our models and apply our models to real-world networks. Conclusions and future work are presented in Section \ref{sec:V}.

\section{RELATED WORK}\label{sec:I}
Watts and Strogatz first proposed small-world networks \cite{WS}. Then, Barab\'{a}si and Albert established the scale-free network model (BA model) \cite{SF}, showing that the growth process and preferential attachment account for the Pareto principle. The two models caused a breakthrough, and a great variety of network models emerge. There are a number of networks changing with time in real life nowadays, and studies on dynamic network models gained more attention. Evolving networks with adding or deleting a vertex by probabilities were constructed \cite{random}. As an extension, an evolution mechanism of weighted edges was taken into consideration \cite{weight}. Under continuous time, stochastic processes were utilized for evolving network modeling. Poisson point processes on space-time were applied to describe interactions between two spaces \cite{bd}, and birth-and-death processes were used for studying a finite tree \cite{tree}. Later, a queueing system was utilized to describe the growth and deletion of vertices of evolving networks \cite{feng1}. Based on this model, different deletion mechanisms were further investigated \cite{feng2}. With big data technology, a generative network framework was established, which preserves both the structural and temporal features of the real data \cite{data}. Dynamic network models with newly-added and deleted nodes, which preserve both heterogeneous and dynamic features of a network were also studied in the field of representation learning method \cite{learning} and mobile ad hoc networks \cite{hoc}. There are other dynamic network models, e.g., temporal networks \cite{tem} and its extension \cite{tem_based}, adaptive networks \cite{ad}, and simplicial activity driven networks (SAD) \cite{sad}. All of these models provide appropriate frameworks for describing dynamic networks.

Network topological properties are also worth studying. There are a lot of traditional methods of calculating degree distributions of networks, e.g., mean-field method and master equations, which are applied under some specific situations to derive degree distributions of evolving networks \cite{dd2} \cite{feng3}. However, in this case of evolving networks with growth and deletion of nodes, it is difficult to obtain an analytical solution due to the non-homogeneity. To address this, a Markov chain method based on stochastic process rules was established and applied to two kinds of evolving networks \cite{zhang}. Later, a reversible Markov chain formulation was proposed to obtain the stationary distribution of degrees \cite{Cihan}. By the Markov chain, an exact formula of degree distributions was obtained in a random evolving network \cite{mar_exa}. Apart from degree distributions, many other topological properties were also studied by researchers, e.g., betweenness \cite{between}, clustering coefficient \cite{cc}, and assortativity \cite{assor} of networks.

The evolving network model proposed in this paper applies stochastic processes which were used in some of the above primer work, however, we do not use the Poisson process to describe the growth and deletion of nodes as described in \cite{feng1}, while we focus on the degree of a vertex and utilize stochastic processes to describe the variance of degree. Additionally, different from the BA model, we use a non-homogenous Poisson process whose parameter is positively correlated to vertices' degree varying with time. In this way, the preferential attachment rule is not described by probabilities but is described from a new perspective by non-homogeneous Poisson rates which are related to the degree of vertices.

\section{\label{sec:II}EVOLVING NETWORKS WITH GROWTH AND REDUCTION}
There are a lot of real-world networks with dynamic features of shrinking as well as enlarging. For instance, in a population contact network, people can come to or leave the network, i.e., the migration of mobile population. Based on the network framework, inflows of the population can be described as the growth of nodes while outflows can be represented by the deletion of nodes represents. To generalize this dynamic feature of real-world networks, we construct an evolving network whose vertices and edges are deleted as well as added by our mechanism.
\subsection{\label{sec:II1}Modeling of Evolving network}
For the growth mechanism of our evolving network, there are newly-coming vertices connecting to existing vertices, and existing vertices with larger degrees have higher probabilities to be connected to new vertices according to the preferential attachment. Different from traditional models, the phenomenon of vertices connecting to new vertices can be regarded as an increasing process of vertex degrees.
Therefore, we assume that the degree of old vertices with larger degrees grows at a higher rate, while the degree of existing vertices with smaller degree values grows at a lower rate. Consequently, we realize the preferential attachment by a degree increase rate positively correlated with the degree.
On the other hand, for the reduction mechanism, there are vertices disconnecting from their neighbors and leaving the network, which occurs randomly. This phenomenon can be regarded as the decrease of degrees of an arbitrary vertex in the network. We thus assume that existing vertices in the network lose degrees randomly at a constant rate.

In respect of the increase process of a vertex's degree, it can also include new connections from old vertices in addition to connecting to newly-coming vertices. The topological properties can be different under the above situation. For the degree increase process only considering connections from new vertices, the network size grows as the degrees of vertices increase. For the degree increase process only considering the change of existing connections in the network, vertices automatically generate new edges to connect existing vertices, and the network will become denser while the network size would not change. Besides, the change of network average degree is different since the network size and the total degree increase together in the former model, while the total degree increases and the network size do not change in the latter one. The situation could be more complicated concerning both processes since there is not only the input of vertices but also connections changing among vertices. For simplicity, we next only discuss the network evolution driven by the newly-coming vertices.

Based on the above assumptions, the growth and reduction of the network can be regarded as the increase and decrease of degrees for each vertex in the network. In that sense, the increase and decrease of each vertex's degree can be described as a queueing system where vertices are customers, coming and leaving the network, and the process of connecting to the existing vertex is regarded as the service process.
The four important components of this queueing system are described in the following.

\textbf{Input process.}
For an arbitrary vertex $i$ in networks, new vertices come and connect to it randomly, and the rate of new vertices connecting to the vertex is proportional to its degree $k$. Hence, the input flow of new vertices follows a Poisson process with a parameter $\lambda(k)$ which is positively correlated to the degree of a vertex with $k$ degrees. We mark this as
\begin{equation}
\lambda(k)\sim k
\end{equation}
Noticeably, the input rate $\lambda(k)$ varies with the change of vertex $i$ degree. Moreover, suppose that the degree of vertex $i$ is $k(t)$ at time $t$, we denote
\begin{equation}
P(k(t+\Delta t)-k(t)=\Delta k)=\frac{(\lambda(k)\Delta t)^{k}}{k!}e^{-\lambda(k)\Delta t}
\end{equation}
as the probability that $\Delta k$ vertices come and connect to vertex $i$ during time $\Delta t$ with the input degree rate $\lambda(k)$

\textbf{Service process.}
Any vertex can be connected to as many vertices as possible, indicating that the number of vertices under service is without limit, numerically from $0$ to $+\infty$. Therefore, the number of servers is infinite, and all servers work in parallel. Besides, the decrease of vertices is random, leading to a Poisson output flow. Hence, the service time under each server is exponentially distributed with a parameter $\mu$ which is the output rate, expressed as
\begin{equation}
P\{T\leq t\}=1-e^{-\mu t},
\end{equation}
where the variable $T$ is the time that a vertex keeps connecting to vertex $i$ not being deleted, and $P\{T\leq t\}$ indicates the probability that the staying time $T$ in the network of a vertex is smaller than $t$.

\textbf{System capacity.}
Since there is no limit for the number of vertices in our proposed model, the vertex degree queueing system capacity is $+\infty$.

\textbf{Queueing discipline.}
Once a new vertex comes, it immediately connects to the vertex $i$. Hence, the waiting time is avoided and the system follows the rule that first comes, first serves.

The above components of the vertex degree queueing system present how vertices are connected by new vertices and how vertices are deleted. According to the input process, the larger the degree is, the higher the rate the vertex obtains new degrees, which guarantees ``the rich gets richer'' rule. For a better illustration, Fig. \ref{fig:constr} presents a vertex degree queueing system.

The increase and decrease of each vertex's degree drive the network evolution, resulting in the growth and reduction of the whole network.
We also briefly demonstrate the evolution of the network as follows.
\noindent\rule[0.15\baselineskip]{8.8cm}{0.5pt}
$\textbf{Evolving Network Modeling}$\\
\noindent\rule[0.05\baselineskip]{8.8cm}{0.5pt}
\begin{figure}[t]
\includegraphics[scale=0.718]{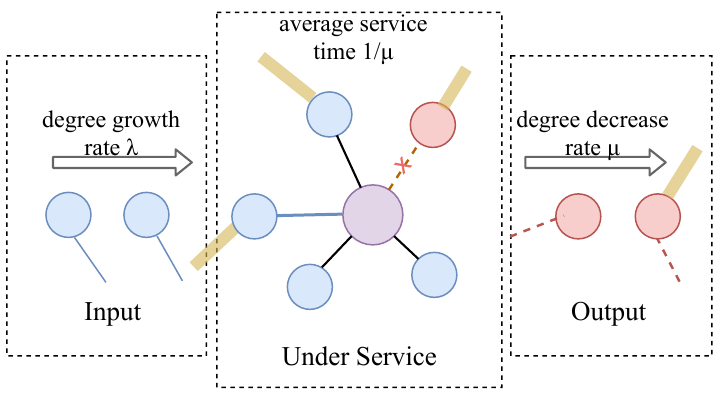}
\centering
\caption{An illustration of the degree queueing system. A snapshot during a short time $t$ is displayed (The degree queueing system evolves with time). The purple larger circle indicates vertex $i$ which is the objective of this system. Blue circles indicate vertices that are connecting or are going to connect to vertex $i$, and those with blue lines in the middle square indicate edges connected during time $t$, while black lines indicate edges connected before time $t$. Red dot lines and red circles indicate edges and vertices disconnected to vertex $i$. Yellow wide lines indicate that there are edges connected to other vertices in the network, not denoting specific edges.}
\label{fig:constr}
\end{figure}

\textbf{\textit{Initial network}}
There are a few $m_{0}$ vertices in the initial network, connecting to each other randomly.

\textbf{\textit{Growth of network.}}
Based on the vertex degree queueing system, each existing vertex follows the Poisson process. Therefore, in light of the additive property of Poisson processes, the input flow of vertices of the whole network follows a Poisson process with a parameter that is the sum of Poisson parameters of all separate Poisson processes of the degree of each vertex.

\textbf{\textit{Reduction of network.}}
According to the vertex degree queueing system, the degree decrease of each vertex follows a Poisson process, which indicates that the edges between vertices are randomly disconnected. For a vertex, if there is no edge between it and other vertices, it is deleted from the network.

\textbf{\textit{Connection and disconnection.}}
For each vertex, new vertices with an edge connect to it. Besides, since any existing vertex (denoted as vertex $i$) with $k$ degrees in the network is connected to new vertices at the rate $\lambda(k)$, which is in a positive correlation with the degree of vertex $i$, the connection mechanism follows the preferential attachment that is
\begin{equation}
P_i=\frac{\lambda_i(k)}{\sum_i\lambda_i(k)}
\end{equation}
where $P_i$ is the probability that vertex $i$ is connected to a newly coming vertex.
For the disconnection, in terms of vertex $i$, vertices leave it once the service process ends, which leads to the disconnection of vertices in the network.

\textbf{\textit{Termination.}}
We set a time $T$ large enough for the termination. Once time reaches $T$, the network ends up evolving.

\noindent\rule[0.05\baselineskip]{9cm}{0.5pt}

The above description indicates a transition from the degree increase and decrease mechanism to the network growth and reduction.
The degree increasing of vertices makes the network grow by the newly coming vertices with edges, while the degree decreasing of vertices makes the network reduce by the removal of vertices and edges. Therefore, the network can be regarded as a multiserver multiqueue queueing system, where each vertex is a server and degrees are customers.
The input flow of this network queueing system includes non-homogeneous Poisson processes with rate $\lambda_i$ of server $i$, and the output flow with $\mu$ is a Poisson process which is the addition of the output of all servers satisfying $\mu=\sum_i^n\mu_i$, where $\mu_i$ is the output rate of server $i$. The stationary condition of our network queueing can be analyzed based on the method in \cite{nonhq}. The network queueing system has a unique stationery if the number of customers $\{N(t),t\geq0\}$ is positive recurrent, following
$\pi diag(\lambda_1, \lambda_2,.., \lambda_n)e<\pi G_ne$, where $\pi$ is the stationary probability vector of $H$, $G_n$ is the matrix of output rate in Eq. 3 in \cite{nonhq}, and $e$ is a column vector with all its elements equal to 1.

From the perspective of the whole network, a newly-coming vertex may connect to more than one existing vertex, hence, it is a one-to-many service process in terms of the whole network queueing system. We suppose new vertices carry $m$ edges and connect to existing ones. Once the degree of an existing vertex in the network increases by 1, it suggests that a new vertex comes to the network and connects to existing vertices, then the total degree of the whole network increases by $2m$. Similarly, once the degree of an existing vertex decreases by 1, it suggests that a vertex which is its neighbor disconnects from the existing one, leading to the total degree of the network decreasing by 2. The total degree of the whole network can be taken as a stochastic process denoted by $\{D(t), t\geq 0\}$. Since the degree of a vertex $i$ increases at the rate $\lambda_i$, the increase process of the total degree is a compound non-homogenous Poisson flow with the parameter $2\sum_i\lambda_i m$ according to the Poisson additive property. In the same way, the decrease process of the total degree follows a Poisson process with the parameter $2\mu$.

\subsection{\label{sec:II1}Analyses for Degree Distribution}
In this section, we perform analyses for the degree distribution of the proposed evolving network.

According to the model above, the degree increase rate $\lambda(k)$ of a vertex with $k$ degrees is not fixed with time-varying, positively correlated with the degree of the vertex, which is denoted by $\lambda(k)\sim k$. We hereby discuss two exact expressions for $\lambda(k)$: $\lambda(k)=k$ and $\lambda(k)=ln(1+k)$. The two expressions for $\lambda$ are both in a positive correlation with vertex degrees, and $\lambda(k)=k$ is the simplest expression to describe ``the rich gets richer'' rule for the increase of degree. $\lambda(k)=ln(1+k)$ also ensures that vertices with larger degrees obtain degrees at a higher rate, which though grows more slowly when $k$ gets larger. Before deducing the degree distribution in these two situations, some assumptions are introduced.

We analyze the degree of an arbitrary vertex $i$ in the network, then the degree distribution of vertex $i$ we deduced is the degree distribution of the network.
The degrees of a vertex $i$ change with time, and the degree value can be taken as a stochastic process, denoted by $\{K(t),\ t\geq0\}$. Furthermore, we assume that the future state of the degree value $K(t+\Delta t)$, given the past states and the present state $K(t)$, is independent of the past states and depends only on the present state, which is intuitive according to the increase and decrease mechanism for the vertex degree. Thus, the degrees of an arbitrary vertex $i$ varying with time denoted by $\{K(t),\ t\geq0\}$ is known as a Markov chain whose state space is $\Omega=\{1,\ 2,\ \cdots\}$, and we let
\begin{equation}
P_k(t)=P\{K(t)=k\},\quad k\in \Omega,
\end{equation}
which indicates the probability that the degree is $k$ at time $t$. Given $t\rightarrow\infty$, if the limited distribution exists, then define it as the stationary distribution
\begin{equation}
P_k=\lim_{t\rightarrow\infty}P_k(t),\quad k\in \Omega.
\end{equation}
that is also the degree distribution of the network if the process of the degree will be stationary when $t\rightarrow\infty$.

Since the vertex degrees increase and decrease, the birth and death process is introduced to deduce expressions for degree distributions. An illustration of the state transition of the birth and death process for the vertex degree is shown in Fig. \ref{fig:bad}. Since the time of each new vertex keeps connecting following an exponential distribution with $\mu$ and there are $k$ vertices connected to a vertex at the same time. Hence, $k\mu$ indicates the total degree decrease rate. Next, we analyze two situations that are $\lambda(k)=k$ and $\lambda(k)=ln(1+k)$ respectively.

\subsubsection{Degree distribution with $\lambda(k)=k$}
We first consider the form $\lambda(k)=k$. Primarily, we give the condition of the existence of the stationary distribution.

\begin{Lemma}\label{lm:1}
Suppose that the degree increase rate $\lambda(k)=k$, the decrease rate is $\mu$ ($\mu$>0), then the stationary degree distribution exists if $\mu>1$.
\end{Lemma}

\begin{proof}
As the queueing system is stationary, according to the Kolmogorov backward equation, for $k$=1, 2, $\cdots$, we have
\begin{equation}\label{eq:k}
(k-1)P_{k-1}-k\mu P_k=0
\end{equation}
Processing Eq. \ref{eq:k}, put $P_{k-1}$ and $P_k$ on both sides of the equation sign, and iterate each equation based on the former equation in Eq. \ref{eq:k}, we obtain $P_k$ ($k=1, 2, \cdots$) expressed by $P_1$,
\begin{equation}\label{eq:k2}
P_k=\frac{1}{k\mu^{k-1}}P_1
\end{equation}
According to the sum of probabilities being 1, sum up the probabilities of all the possible values of $k$, denoted by $\sum_{k=1}^{+\infty}P_k=1$.
In the light of Eq. \ref{eq:k2}, it becomes
\begin{equation}\label{eq:p1}
P_1(1+\sum_{k=2}^{+\infty}\frac{1}{k\mu^{k-1}})=1
\end{equation}
Let $S(\mu)=\sum_{k=2}^{+\infty}\frac{1}{k\mu^{k-1}}$, and denote by $u_k$ the term $k$ in the series $S(\mu)$.
If $S(\mu)$ converges, then the expression for $P_k$ exists according to Eqs. \ref{eq:k2} and \ref{eq:p1}.

For $S(\mu)$, if $\lim_{k\rightarrow+\infty}\frac{u_k}{u_{k-1}}$<1, $S(\mu)$ converges. We have
\begin{equation}
\lim_{k\rightarrow+\infty}\frac{u_k}{u_{k-1}}=\lim_{k\rightarrow+\infty}\frac{1}{k{\mu^{k-1}}}\cdot(k-1){\mu^{k-2}}=\lim_{k\rightarrow+\infty}\frac{k-1}{k}\cdot\frac{1}{\mu}
\end{equation}
Since $\lim_{k\rightarrow+\infty}\frac{k-1}{k}=1$,
we obtain
\begin{equation}\label{eq:f}
\lim_{k\rightarrow+\infty}\frac{u_k}{u_{k-1}}=\frac{1}{\mu}
\end{equation}
According to Eq. \ref{eq:f}, if $\mu$>1, $u_k$ is larger than $u_{k-1}$, $S(\mu)$ is divergent.
If $\mu$=1, the general term is
\begin{equation}
u_k=\frac{1}{k}
\end{equation}
and $S(\mu)=\sum_{k=2}^{+\infty}\frac{1}{k}$ is a harmonic series which is divergent leading to no solution to $P_1$ in light of Eq. \ref{eq:p1}.
If $\mu>1$, $S(\mu)$ converges, further the expression for the stationary degree distribution $P_{k}$ exists according to Eq. \ref{eq:p1}.

Therefore, the stationary degree distribution $P_k$ exists if $\mu>1$.
The result follows.
\end{proof}

\begin{figure}[t]
\includegraphics[scale=0.718]{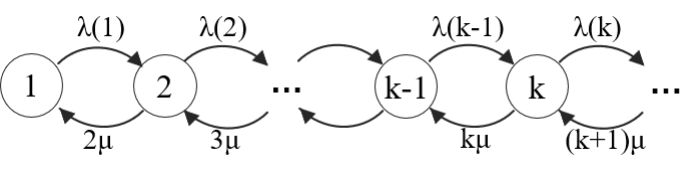}
\centering
\caption{An illustration of the state transition of a vertex (denoted by $i$)'s degree: As we can see above the arrow $f(k)$$(k=1,2,...)$ indicates the degree increase rate relevant to the degree $k$, and beneath the arrow, $k\mu$ indicates the degree decrease rate. The figures in circles indicate the degree value.}
\label{fig:bad}
\end{figure}
Lem. \ref{lm:1} demonstrates the condition for the existence of the stationary distribution for the vertex degree. We know that if $\lambda(k)=k$, $\mu$ must be larger than 1 for the theoretical solution for stationary degree distribution.
Under the condition that $\mu>1$, we next deduce the expression for $P_k$. Before deducing the expression for $P_k$, we introduce Lemma \ref{lm:2} for preparations.
\begin{Lemma}\label{lm:2}
Giving the infinite series $S(x)=\sum_{k=1}^{+\infty}\frac{x^{k-1}}{k}$, and $x<1$, $S(x)$ converges and
\begin{equation}
S(x)=-\frac{ln(1-x)}{x}
\end{equation}
\end{Lemma}
\begin{Proof}
$S(x)$ can be processed to
\begin{equation}\label{eq:s}
S(x)=\frac{1}{x}\sum_{k=1}^{+\infty}\frac{x^k}{k}
\end{equation}
where $x$ is regarded as a constant, not a variable.
Let $\sum_{k=1}^{+\infty}\frac{x^k}{k}$ be $S_1(x)$,
Take the derivative of $S_1(x)$ with respect to $k$, marked as $S^{\prime}_1(x)$, based on the algorithm of power series, we yield
\begin{equation}\label{eq:xn}
S^{\prime}_1(x)=\sum_{k=1}^{+\infty}x^{k-1}=\frac{1}{1-x}
\end{equation}
which utilizes the summation formula of geometric series.

Then, integrate $S^{\prime}_1(x)$, we obtain
\begin{equation}\label{eq:s1}
S_1(x)=\int_{0}^{x}\frac{1}{1-x}dx=-ln(1-x)
\end{equation}
According to Eq. \ref{eq:s} and \ref{eq:s1}, $S(x)$ is
\begin{equation}
S(x)=-\frac{ln(1-x)}{x}
\end{equation}
The result follows.
\end{Proof}

Based on Lemmas. \ref{lm:1} and \ref{lm:2}, the exact expression for the degree stationary distribution $P_k$ is presented below.

\begin{Theorem}\label{th:1}
For an arbitrary vertex $i$ in the proposed network, the degree increase rate $\lambda(k)$ of it at time $t$ is equal to the degree $k(t)$ of vertex $i$, and the degree decrease rate is $\mu$ where $\mu<1$. Then the stationary degree distribution $P_k$ of the vertex $i$ exists and is denoted as
\begin{equation}
P_k=-\frac{1}{k\mu^{k}}ln^{-1}(1-\frac{1}{\mu})
\end{equation}
\end{Theorem}

\begin{Proof}
Referring to Eq. \ref{eq:p1} in Lemma \ref{lm:1}, we obtain $P_1$ expressed as
\begin{equation}\label{eq:pp1}
P_1=\frac{1}{(1+\sum_{k=2}^{+\infty}\frac{1}{k\mu^{k-1}})}
\end{equation}
According to Lem. \ref{lm:2}, the denominator $\sum_{k=2}^{+\infty}\frac{1}{k\mu^{k-1}}$ of the right term of Eq. \ref{eq:pp1} can be written as
\begin{equation}\label{eq:k22}
\sum_{k=2}^{+\infty}\frac{1}{k\mu^{k-1}}=\sum_{k=1}^{+\infty}\frac{x^{k-1}}{k}-u_1=\sum_{k=1}^{+\infty}\frac{x^{k-1}}{k}-1
\end{equation}
where $u_1$ is the first term of the series $S(\mu)$ in Lem. \ref{lm:1}. Taking Eq. \ref{eq:k22} into Eq. \ref{eq:pp1}, then
\begin{equation}
P_1=\frac{1}{\sum_{k=1}^{+\infty}\frac{1}{k\mu^{k-1}}}
\end{equation}
Applying Lemma. \ref{lm:2}, we get $\sum_{k=1}^{+\infty}\frac{1}{k\mu^{k-1}}=-\frac{1}{\mu ln(1-\frac{1}{\mu})}$.
Thus,
\begin{equation}
P_1=-\frac{1}{\mu ln^(1-\frac{1}{\mu})}
\end{equation}

Then, we obtained $P_k$ by iterations according to Eq. \ref{eq:k2},
\begin{equation}
P_k=-\frac{1}{k\mu^{k}}ln^{-1}(1-\frac{1}{\mu})
\end{equation}

The result follows.
\end{Proof}

\begin{Theorem}\label{th:2}
In our proposed model, assume that the degree increase rate of an arbitrary vertex $i$ is equal to its degree, denoted as $\lambda(k)=k$ of vertex $i$, and the degree decrease rate is $\mu$, where $\mu<1$. Then the expectation of the degree of vertex $i$ is
\begin{equation}\label{eq:ex}
E[k]=-\frac{1}{(\mu-1)ln(1-\frac{1}{\mu})}
\end{equation}
\end{Theorem}

\subsubsection{Degree distribution with $\lambda(k)=ln(1+k)$}
We consider an alternative expression for the degree increase rate $\lambda(k)=ln(k+1)$.
We also first analyze the existence of the stationary distribution for $K(t)$ with $f(k)=ln(k+1)$ in Lemma. \ref{lm:3}.

\begin{Lemma}\label{lm:3}
Suppose that the degree increase rate $\lambda(k)$ is equal to $\ln(k+1)$, and the decrease rate is $\mu$ ($\mu$>0). Then the stationary degree distribution exists.
\end{Lemma}

\begin{Proof}
As the queueing system is stationary, according to the Kolmogorov backward equation, we have
\begin{equation}\label{eq:lnk}
ln(1+k)P_{k-1}-k\mu P_k=0
\end{equation}
Put $P_{k-1}$ and $P_k$ on both sides of the equation sign, iterate each equation in Eq. \ref{eq:lnk} based on the former equation, we obtain $P_k$ ($k=1, 2, \cdots)$) expressed by $P_1$,
\begin{equation}\label{eq:lnk2}
P_k=\frac{\prod_{j=1}^{k-1}ln(1+j)}{k!\mu^{k-1}}P_1
\end{equation}
According to the sum of probabilities being 1, sum up the probabilities of all the possible values of $k$, denoted by $\sum_{k=1}^{+\infty}P_k$=1,
Based on Eq. \ref{eq:lnk2}, the sum of probabilities is
\begin{equation}\label{eq:lnp1}
P_1(1+\sum_{k=2}^{+\infty}\frac{\prod_{j=1}^{k-1}ln(1+j)}{k!\mu^{k-1}})=1
\end{equation}
Let $\sum_{k=2}^{+\infty}\frac{\prod_{j=1}^{k-1}ln(1+k)}{k!\mu^{k-1}}$ be the infinite series $S(\mu)$, and denote by $u_k$ the general term of the series $S(\mu)$.
If $\lim_{k\rightarrow+\infty}\frac{u_k}{u_{k-1}}$<1, $S(\mu)$ converges, and the expression for the degree distribution $P_k$ exists.

Test the convergence of $S(\mu)$ via $u_k$ divided by $u_{k-1}$,
\begin{equation}
\begin{aligned}
\lim_{k\rightarrow+\infty}\frac{u_k}{u_{k-1}}&=\lim_{k\rightarrow+\infty}\frac{\prod_{j=1}^{k}ln(1+j)}{k!{\mu^{k-1}}}\cdot\frac{(k-1)!\mu^{k-2}}{\prod_{j=1}^{k-1}ln(1+j)}\\
&=\lim_{k\rightarrow+\infty}\frac{ln(1+k)}{k}\cdot\frac{1}{\mu}
\end{aligned}
\end{equation}
From
\begin{equation}
\lim_{k\rightarrow+\infty}\frac{ln(1+k)}{k}=0
\end{equation}
it follows,
\begin{equation}\label{eq:lnf}
\lim_{k\rightarrow+\infty}\frac{u_k}{u_{k-1}}=0
\end{equation}
$S(\mu)$ converges according to Eq. \ref{eq:lnf}, and the stationary degree distribution exists.

The result follows.
\end{Proof}

Though the expression for the degree stationary distribution $P_k$ can be obtained if $S(\mu)$ in Lemma. \ref{lm:3} is known according to Eq. \ref{eq:lnp1}, the theoretical solution for $S(\mu)$ is difficult to attained.
Therefore, we let the series $S(\mu)$ be
\begin{equation}\label{eq:s}
S(\mu)=\sum_{k=2}^{+\infty}\frac{\prod_{j=1}^{k-1}ln(1+j)}{k!\mu^{k-1}}
\end{equation}
We calculate $S(\mu)$ by numerical simulations. The sum should be calculated from $k=2$ to $+\infty$ theoretically. However, in simulations, we set $k$ large enough, meanwhile the largest value of $k$ should avoid the underflow of items in $S(\mu)$. Then, set the value for $\mu$, by computer calculation we can obtain $S(\mu)$ numerically, and further get $P_1$ according to
\begin{equation}
P_1=\frac{1}{1+S(\mu)}
\end{equation}
Once obtaining $P_1$, we can get $P_k$ ($k=1,2,\cdots$) by iterations based on Eq. \ref{eq:lnk2}. The plots of the degree distribution with $\lambda(k)=ln(1+k)$ will be demonstrated in the simulation.

Above all, we construct the evolving network with growth and decrease of vertices and edges, and perform analyses for the degree distribution via the birth and death process. More properties of the proposed network are further presented in the next section.

\begin{figure*}[t]
\begin{minipage}[t]{0.33\linewidth}
\centering
 \includegraphics[scale=0.37]{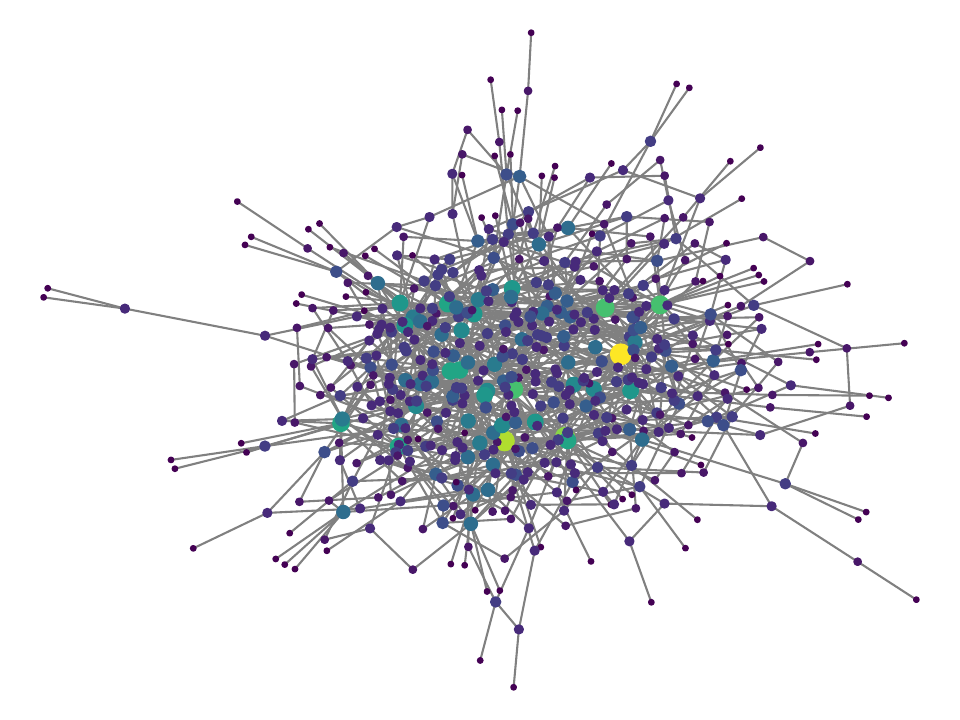}
\parbox{4.5cm}{\centering\footnotesize (a) Network with $\lambda(k)$=$k$ and $\mu$=1.01}

\end{minipage}
\begin{minipage}[t]{0.33\linewidth}
\centering
 \includegraphics[scale=0.46]{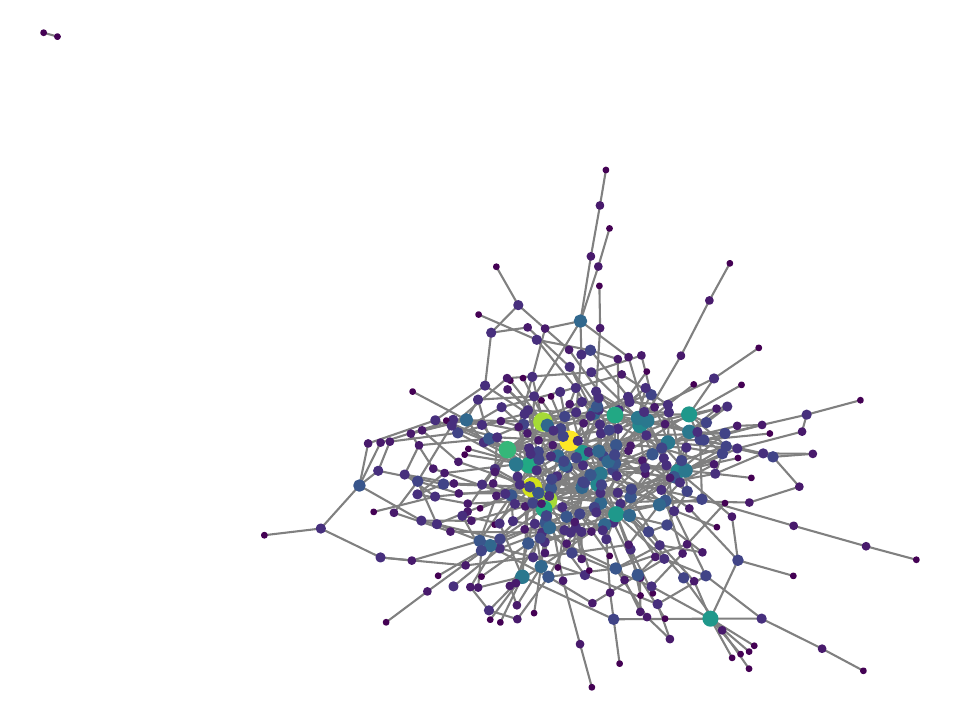}
\parbox{4.5cm}{\centering\footnotesize (b) Network with $\lambda(k)$=$k$ and $\mu$=1.21}

\end{minipage}
\begin{minipage}[t]{0.33\linewidth}
\centering
 \includegraphics[scale=0.48]{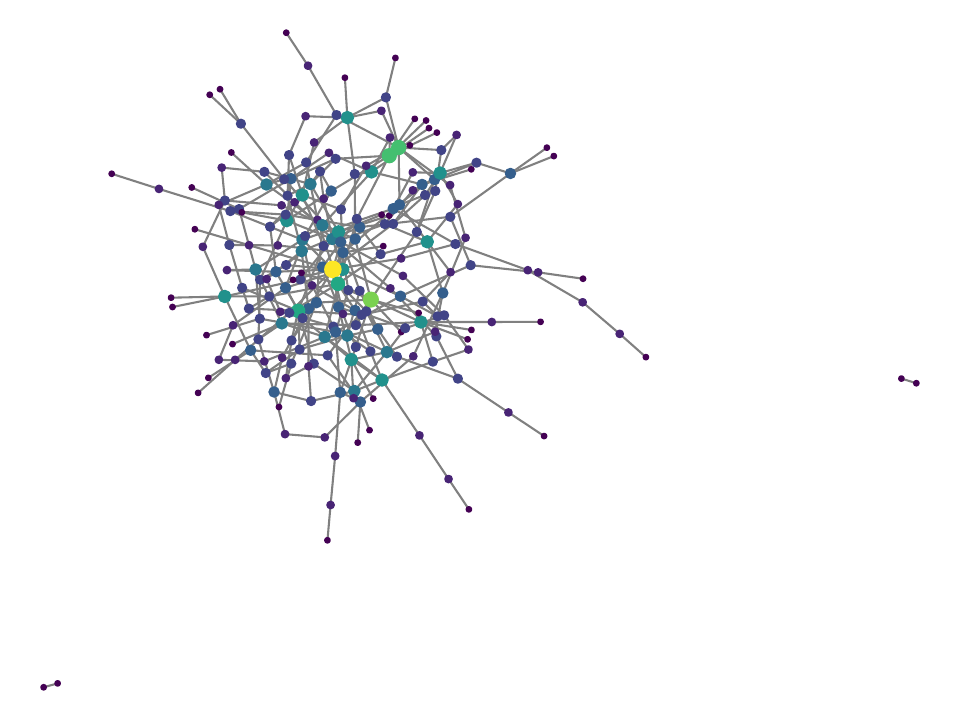}
\parbox{4.5cm}{\centering\footnotesize (c) Network with $\lambda(k)$=$k$ and $\mu$=1.41}

\end{minipage}
\begin{minipage}[t]{0.33\linewidth}
\centering
 \includegraphics[scale=0.35]{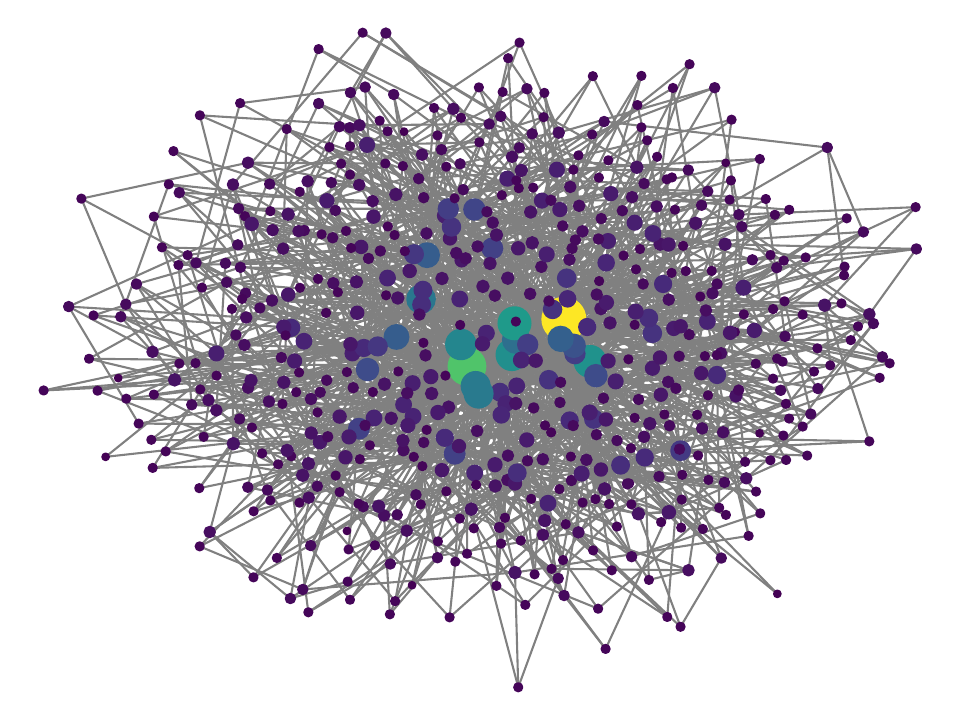}
\parbox{4.5cm}{\centering\footnotesize (d) Network with $\lambda(k)$=$ln(1+k)$ and $\mu$=0.005}

\end{minipage}
\begin{minipage}[t]{0.33\linewidth}
\centering
 \includegraphics[scale=0.35]{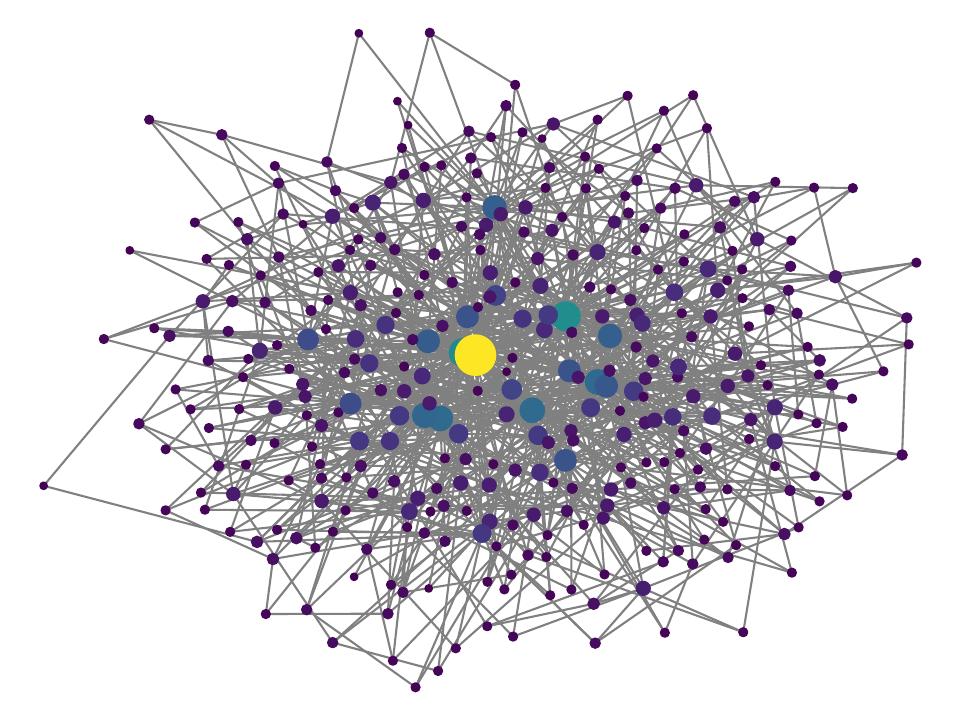}
\parbox{4.5cm}{\centering\footnotesize (e) Network with $\lambda(k)$=$ln(1+k)$ and $\mu$=0.01}

\end{minipage}
\begin{minipage}[t]{0.33\linewidth}
\centering
 \includegraphics[scale=0.35]{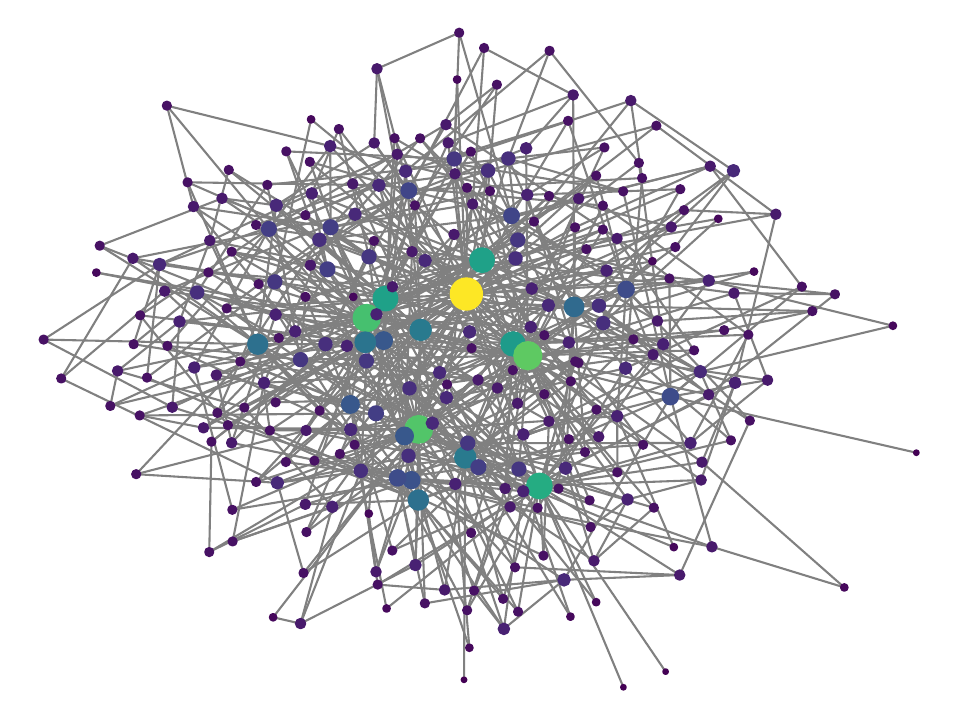}
\parbox{4.5cm}{\centering\footnotesize (f) Network with $\lambda(k)$=$ln(1+k)$ and $\mu$=0.05}

\end{minipage}
\caption{\label{fig:network}Illustration of networks under two degree increase rates and different decrease rates: The initial network is a scale-free network with 10 vertices. New vertices come with two edges connected to existing vertices. Sub-figures (a), (b), and (c) respectively show snapshots of networks with $\mu$=1.01, 1.21, and 1.41 under the first situation that $\lambda(k)=k$, while sub-figures (d), (e), and (f) respectively present snapshots of networks with $\mu$=0.005, 0.01, and 0.05 under the second situation that $\lambda(k)=ln(1+k)$. The size and the color of nodes symbolize their degree value, where big and yellow circles indicate a large degree while small and purple circles indicate a small degree value. Vertices with a medium degree value are middle-size and dark blue.}
\end{figure*}

\begin{figure*}[h!]
\begin{minipage}[t]{0.33\linewidth}
\centering
 \includegraphics[scale=0.37]{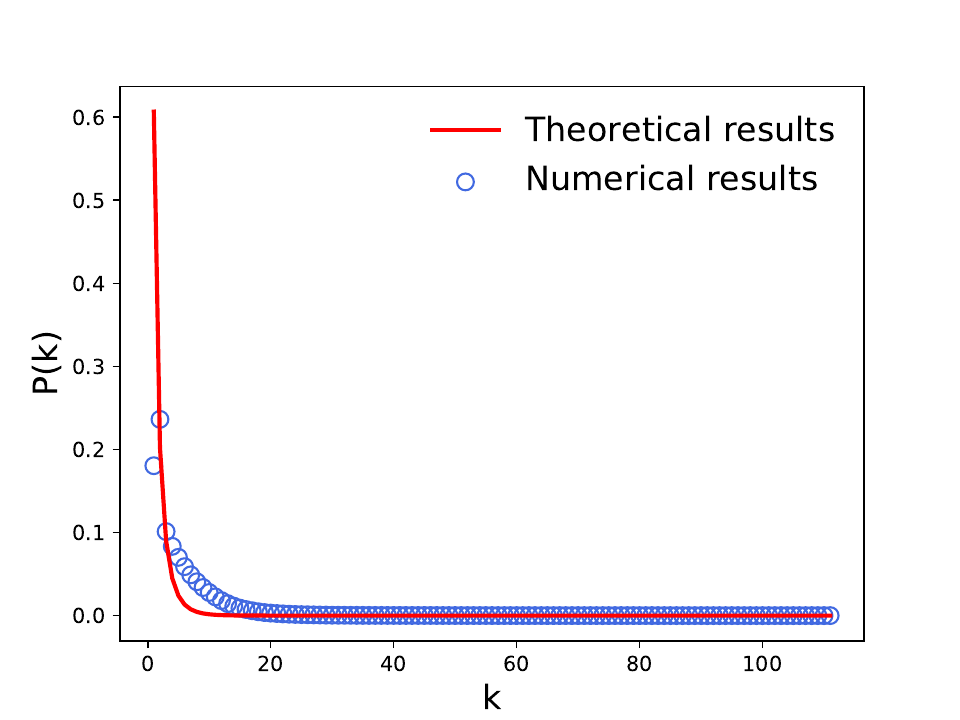}
\parbox{4.5cm}{\centering\footnotesize (a) The degree distribution with $\mu$=1.5}

\end{minipage}
\begin{minipage}[t]{0.33\linewidth}
\centering
 \includegraphics[scale=0.37]{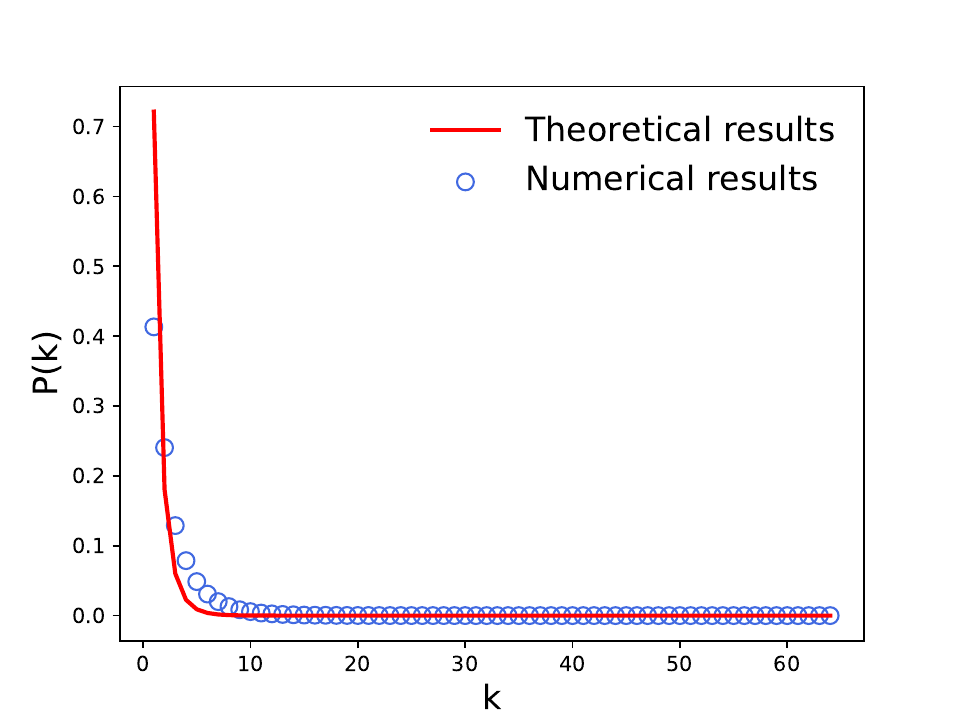}
\parbox{4.5cm}{\centering\footnotesize (b) The degree distribution with $\mu$=2.0}

\end{minipage}
\begin{minipage}[t]{0.33\linewidth}
\centering
 \includegraphics[scale=0.37]{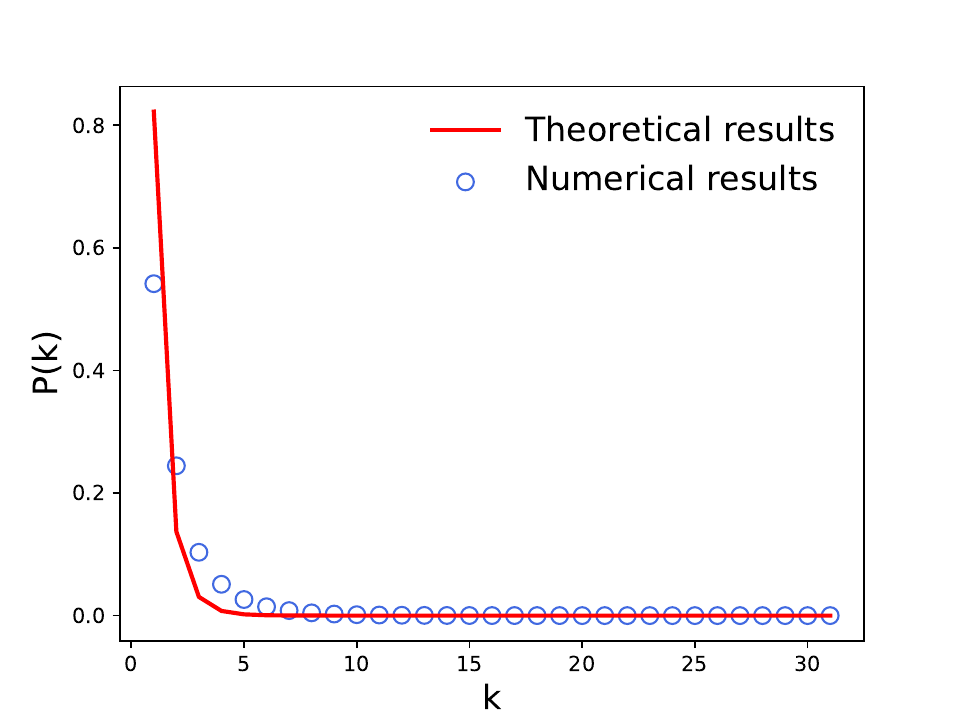}
\parbox{4.5cm}{\centering\footnotesize (c) The degree distribution with $\mu$=2.5}

\end{minipage}
\begin{minipage}[t]{0.33\linewidth}
\centering
 \includegraphics[scale=0.37]{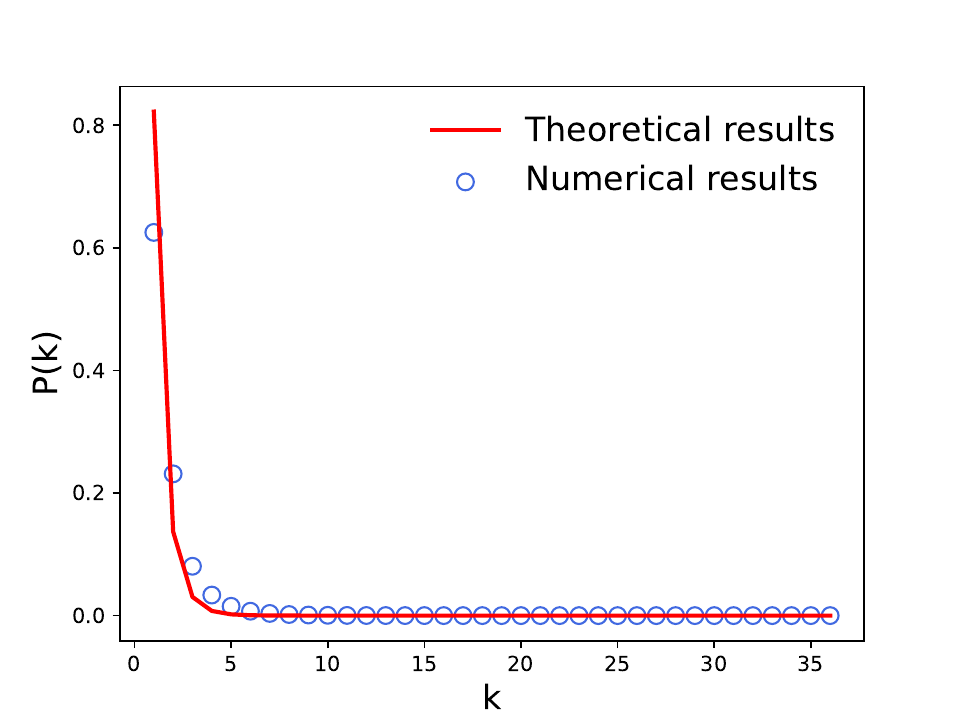}
\parbox{4.5cm}{\centering\footnotesize (d) The degree distribution with $\mu$=3.0}

\end{minipage}
\begin{minipage}[t]{0.33\linewidth}
\centering
 \includegraphics[scale=0.37]{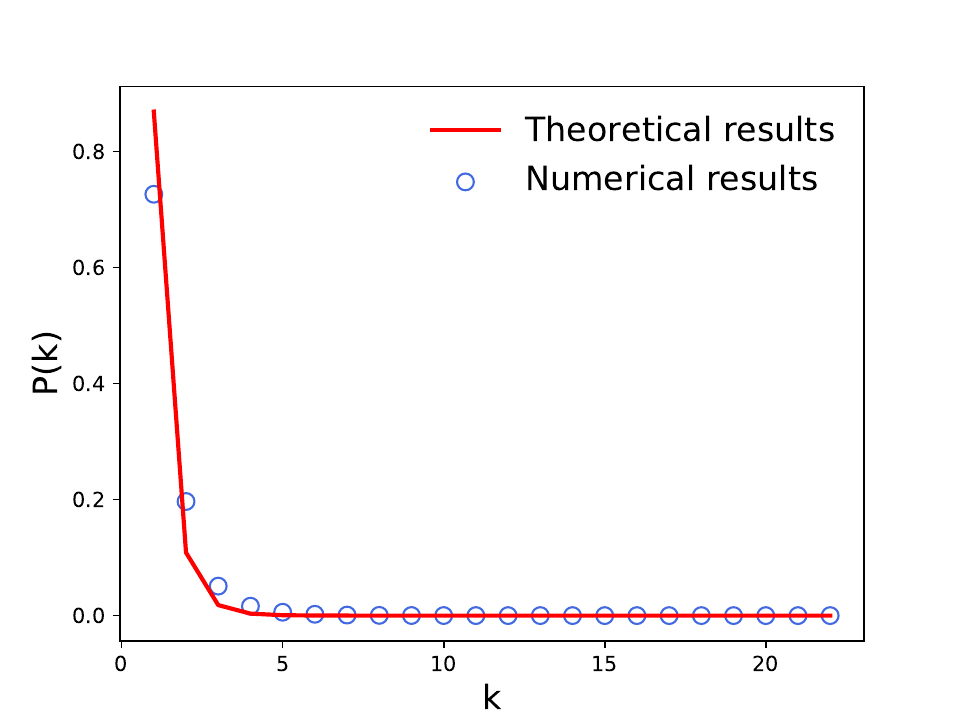}
\parbox{4.5cm}{\centering\footnotesize (e) The degree distribution with $\mu$=4.0}

\end{minipage}
\begin{minipage}[t]{0.33\linewidth}
\centering
 \includegraphics[scale=0.37]{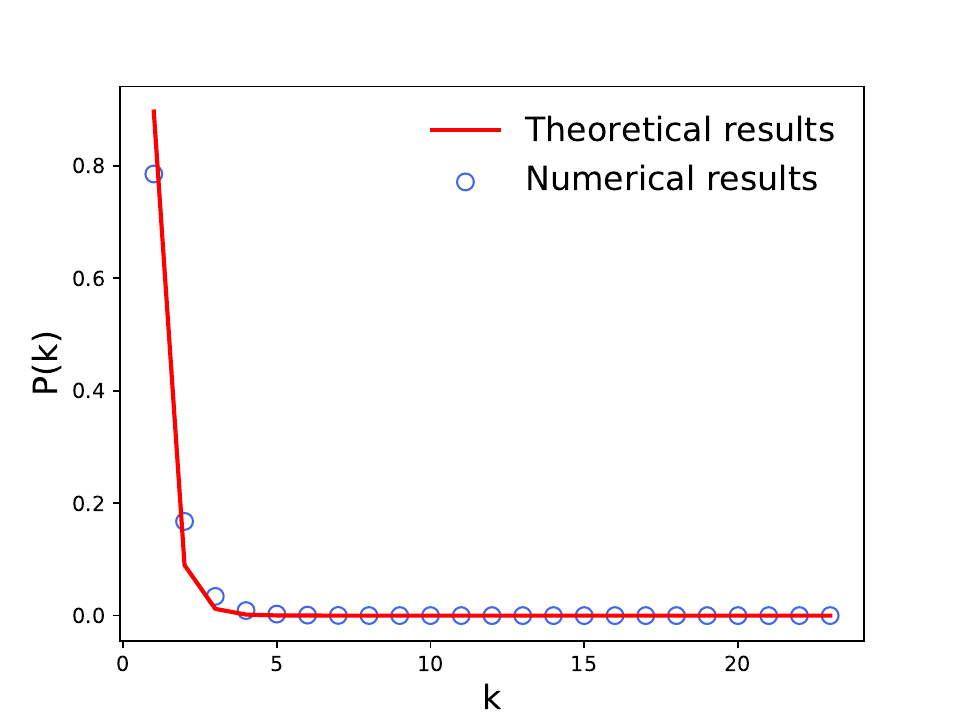}
\parbox{4.5cm}{\centering\footnotesize (f) The degree distribution with $\mu$=5.0}

\end{minipage}
\caption{\label{fig:dist1}The degree distribution with different decrease rates $\mu$ and $\lambda(k)$=$k$: Sub-figures (a)-(f) are respectively set with $\mu$=1.5, 2.0, 3.5, 3.0, 4.0, and 5.0. Time is set large enough for the distribution to be stationary. Red lines indicate theoretical distributions, and the blue circle plots indicate distributions obtained by simulations. Distributions are all heterogeneous, featured with ``long tail''.}
\end{figure*}

\section{\label{sec:III}Simulation}
Hereby, we display simulation results to verify and extend our theorems of the proposed model, and also apply our theorems to fitting real data, showing the feasibility of our model.

\begin{figure*}[t]
\begin{minipage}[t]{0.33\linewidth}
\centering
 \includegraphics[scale=0.37]{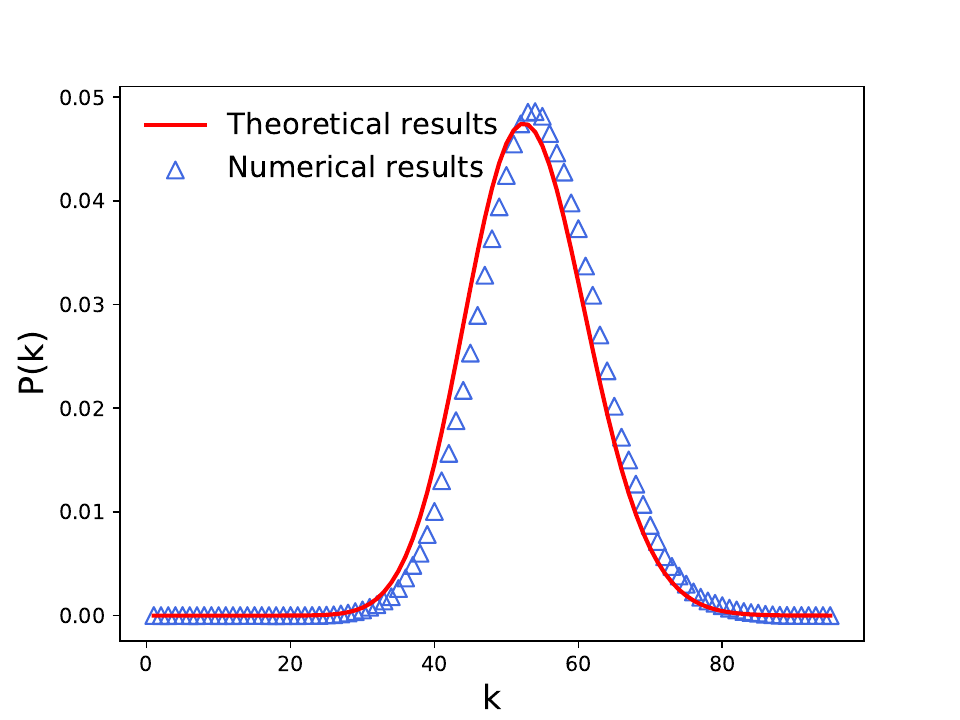}
\parbox{4.5cm}{\centering\footnotesize (a) The degree distribution with $\mu$=0.075}

\end{minipage}
\begin{minipage}[t]{0.33\linewidth}
\centering
 \includegraphics[scale=0.37]{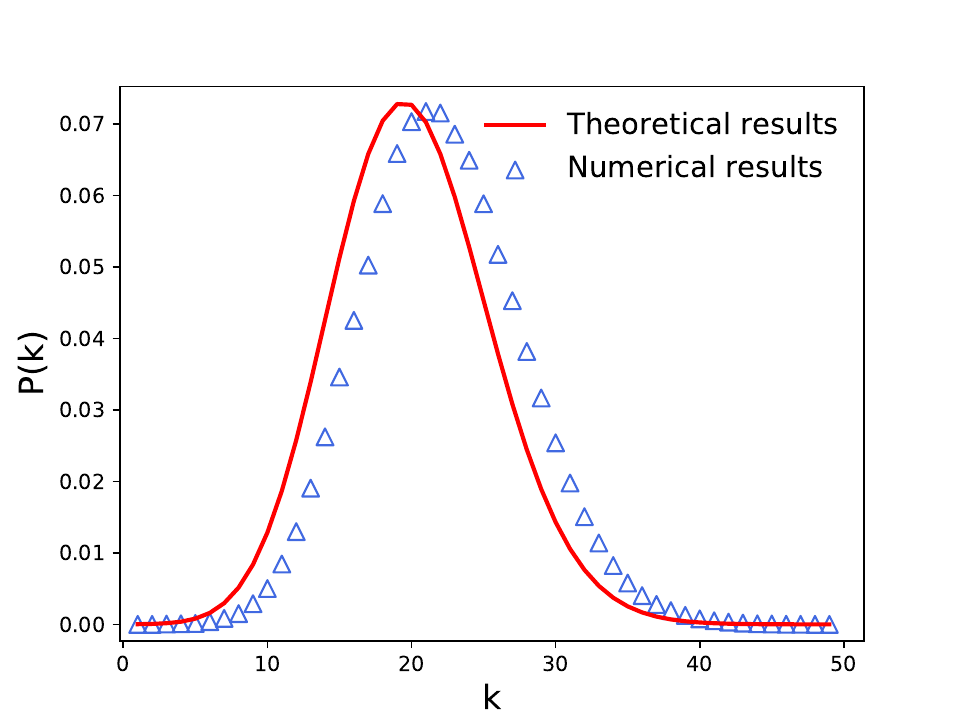}
\parbox{4.5cm}{\centering\footnotesize (b) The degree distribution with $\mu$=0.15}

\end{minipage}
\begin{minipage}[t]{0.33\linewidth}
\centering
 \includegraphics[scale=0.37]{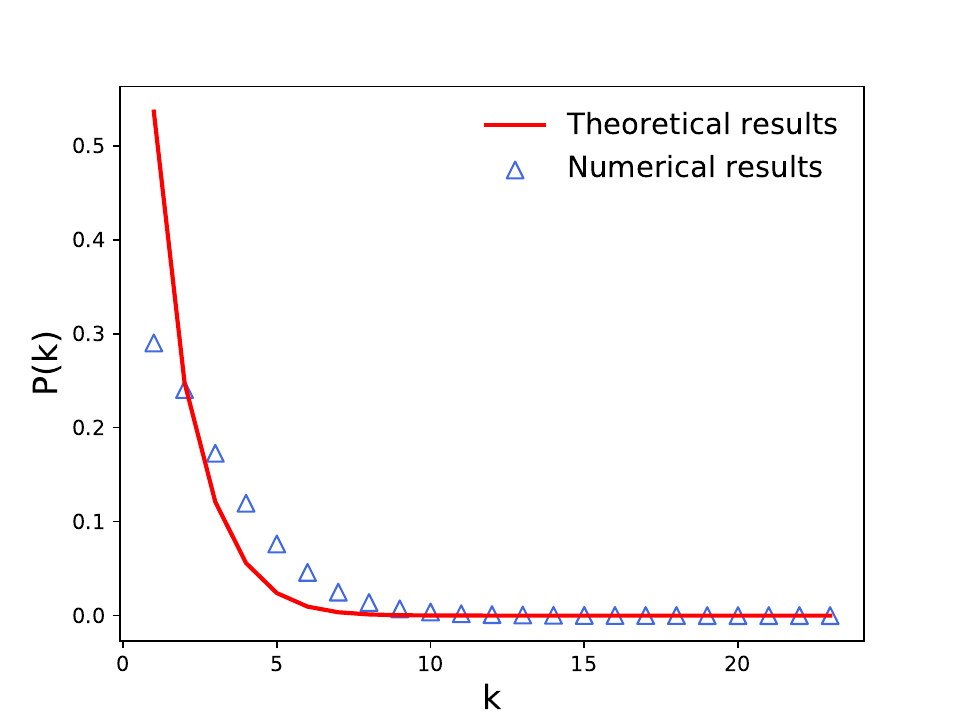}
\parbox{4.5cm}{\centering\footnotesize (c) The degree distribution with $\mu$=0.75}

\end{minipage}
\begin{minipage}[t]{0.33\linewidth}
\centering
 \includegraphics[scale=0.37]{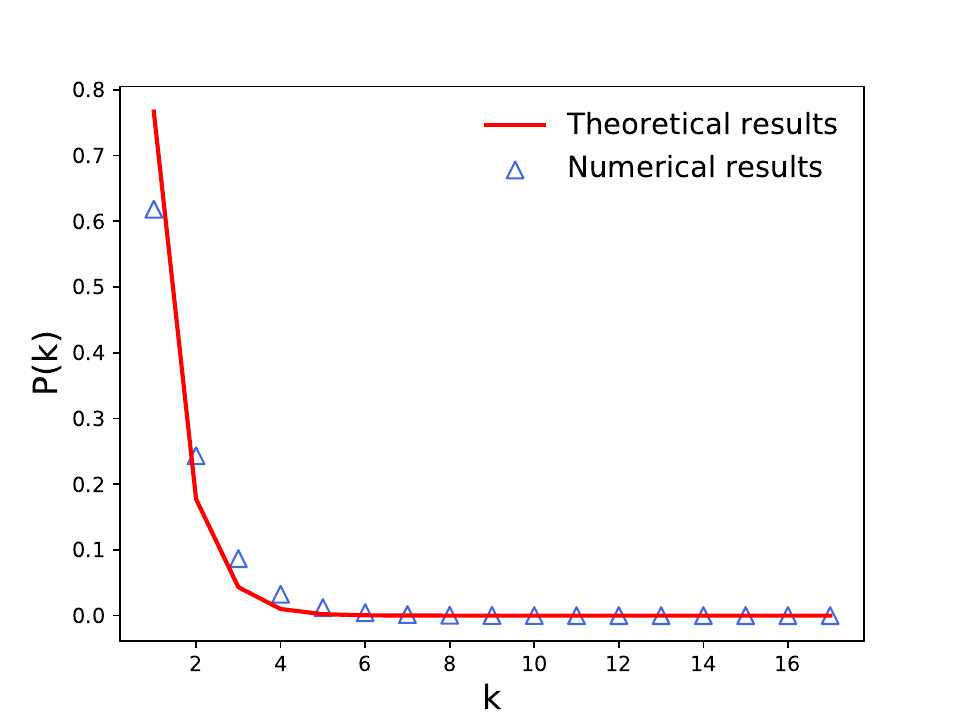}
\parbox{4.5cm}{\centering\footnotesize (d) The degree distribution with $\mu$=1.5}

\end{minipage}
\begin{minipage}[t]{0.33\linewidth}
\centering
 \includegraphics[scale=0.37]{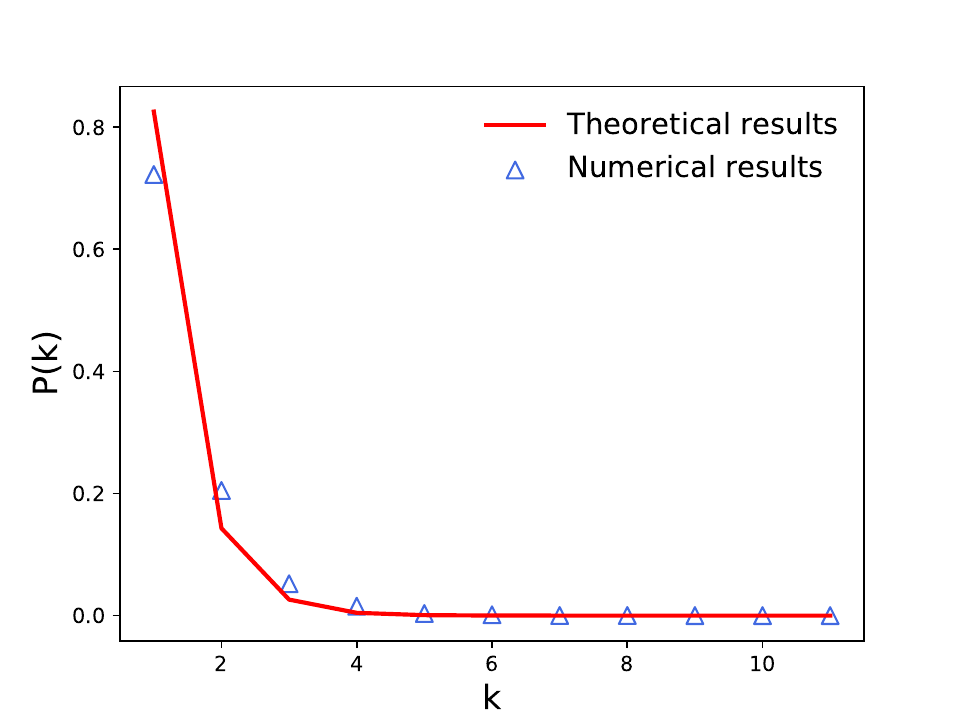}
\parbox{4.5cm}{\centering\footnotesize (e) The degree distribution with $\mu$=2.0}

\end{minipage}
\begin{minipage}[t]{0.33\linewidth}
\centering
 \includegraphics[scale=0.37]{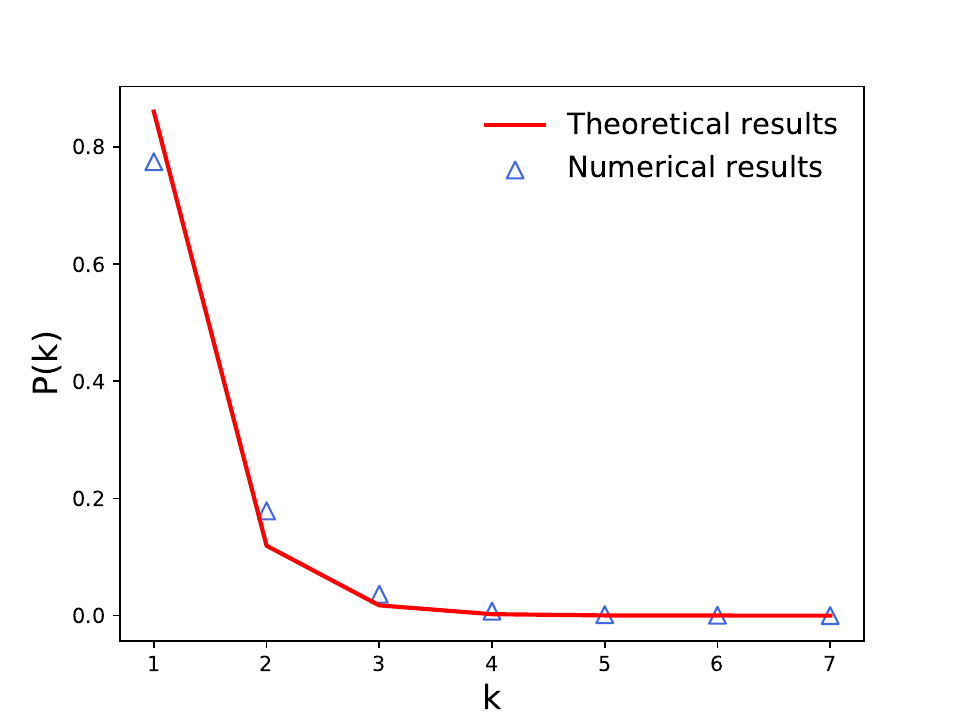}
\parbox{4.5cm}{\centering\footnotesize (f) The degree distribution with $\mu$=2.5}

\end{minipage}
\caption{\label{fig:dist2}The degree distribution with different decrease rates $\mu$ and $\lambda(k)$=$ln(1+k)$: Sub-figures (a)-(f) are respectively set with $\mu$=0.075, 0.15, 0.75, 1.5, 2.0, and 2.5. Red lines indicate theoretical distributions by numerical calculation, and the blue triangle plots indicate distributions obtained by simulations. Distributions with $\mu$=0.075 and 0.15 different from the rest distributions are skew distributions, while others are heterogeneously featured with ``long tail''. }
\end{figure*}


\subsection{Modeling of Network}
Primarily, we describe the construction of our network in simulations based on our model.
In terms of time passing, the time is continuous, and the interval of events that new vertices come and existing vertices are deleted is generated by exponential distributions.
For the connection, each new vertex carried with $m$ edges comes and connects to existing vertices in networks. Allocate new edges of newly coming vertices connecting to existing vertices according to the form of the degree increase rate $\lambda(k)$. For the first form $\lambda(k)$=$k$, we use the $\frac{k_i}{\sum^j k}$ as the probability for vertex $i$ connected to a new vertex. For the second form $\lambda(k)$=$ln(1+k)$, $\frac{ln(1+k_i)}{\sum^j ln(1+k_j)}$ is taken as the probability for vertex $i$ connected to a new vertex.
We utilize $networkx$ packages in Python to carry out the simulation. The procedures of network construction are as follows:

1) Give the termination time $T$, the initial number of vertices $m_0$, and the edge number $m$ of a newly coming vertex carrying.

2) Generate an interval $\Delta t_{\lambda(k)}$ following exponential distribution for the event of vertex increase and another interval $\Delta t_{\mu}$ for the event of edge decrease. If $\Delta t_{\lambda(k)}<\Delta t_{\mu}$, generate a random number in the range [0,1] by $rand()$ multiplied by the sum of degree, and find the corresponding vertex whose degree interval includes the generated number, otherwise, randomly delete an edge in the network.

3) Update the time left.

4) Return 2) until achieving the termination $T$.

Based on the above algorithm, we construct networks demonstrated in Fig. \ref{fig:network}. The sub-figures (a)-(c) show network snapshots with $\lambda(k)=k$, given $\mu$=1.01, 1.21, 1.41 respectively, and sub-figures (d)-(f) present network snapshots with $\lambda(k)=ln(1+k)$, given $\mu$=0.005, 0.01, 0.05. The initial number of vertices is set as $m_0$=10, and the edges of new vertices are $m$=3 for both forms of $\lambda(k)$. $T$ is set as 3.5 for networks in Fig. \ref{fig:network} (a)-(c) and 6.5 for (d)-(f), which is a reasonable length of time for network evolving in terms of the decrease rate $\mu$.
The size and color of a vertex reflect its degree value. In terms of size, the bigger the size of the circle is, the larger degree it has. In terms of color, yellow indicates a large degree, while the other purple indicates a small degree, blue and green indicate a medium degree. The measure of color presenting the degree is only within the same network.
As is shown in Fig. \ref{fig:network}, networks in the sub-figures are all heterogenous, only a few vertices have a large degree value and the majority of vertices have small degree values.
For both situations for the increase rate $\lambda(k)$ ($\lambda(k)$=$k$ and $\lambda(k)$=$ln(1+k)$), the network is getting sparser with the decrease rate $\mu$ increasing according to the snapshots of networks at the same time.

\subsection{Degree Distribution of Limited Resources network}\label{sec:IIIB}
In this sub-section, we focus on the degree distributions according to our vertex degree queueing system.

The degree distributions based on our model are displayed under different situations given different parameters.
For $\lambda(k)=k$, we let $T$ be 10$^5$ large enough for results to satisfy the law of large numbers, and set the degree decrease rate $\mu$=1.5, 2.0, 2.5, 3.0, 4.0, and 5.0. We record the degree value varying with time, calculate the frequency of each degree value appearing, then respectively demonstrate the theoretical degree distributions and simulation degree distribution in sub-figures \ref{fig:dist1}(a)-(f). As we can see in Fig. \ref{fig:dist1}, simulation distribution plots are close to theoretical distributions, which verifies the validity of Thm. \ref{th:1}. Besides, the degree distributions are all heterogenous with a long tail. This heterogenous distribution is due to our mechanism for the increase rate $\lambda(k)=k$, which ensures that the larger degree a vertex has, the quicker its degree grows. Additionally, according to the probabilities on the y-axis from sub-figures \ref{fig:dist1} (a)-(f), the probability of a small degree gets higher as $\mu$ increases, which suggests that the degree distribution becomes more heterogenous with $\mu$ rising. The reason for this is that we can only record degree values within a limited time though it is quite long. The large degree is more difficult to accumulate when the decrease rate $\mu$ becomes larger, leading to the probability of small degree values being higher.
\begin{table}[htbp]
  \centering
  \caption{Similarity of theoretical distributions and simulation distributions based on three measures}
    \begin{tabular}{cccc|cccc}
    \hline
    \hline
          & \multicolumn{3}{c|}{$\lambda(k)=k$} &       & \multicolumn{3}{c}{$\lambda(k)=ln(1+k)$} \\
\cline{1-8}   $\mu$     &$\rho$     & KL    & JS    &$\mu$     &$\rho$     & KL    & JS \\
    1.5   & 0.809 & 0.812 & 0.140 & 0.075 & 0.987 & 0.019 & 0.005 \\
    2     & 0.920 & 0.442 & 0.076 & 0.15  & 0.945 & 0.070 & 0.017 \\
    2.5   & 0.964 & 0.229 & 0.043 & 0.75  & 0.910 & 0.268 & 0.053 \\
    3     & 0.978 & 0.155 & 0.003  & 1.5   & 0.985 & 0.084 & 0.017 \\
    4     & 0.989 & 0.090 & 0.018 & 2     & 0.994 & 0.042 & 0.009 \\
    5     & 0.993 & 0.062 & 0.013 & 2.5   & 0.996 & 0.031 & 0.007 \\
    \hline
    \hline
    \end{tabular}%
  \label{tab:0}%
\end{table}%

For $\lambda(k)=ln(1+k)$, we set $\mu$=0.075, 0.15, 0.75, 1.5, 2.0, and 2.5 and $T=10^5$. We also calculate the frequency of each degree value as their probability according to the law of large numbers. As is shown in Figs. \ref{fig:dist2}(a)-(f), the simulation results are close to the theoretical results, presenting the validation of our theoretical distributions. Besides, the degree distributions with $\lambda(k)=ln(1+k)$ are different from that with $\lambda(k)=k$. As is illustrated in Fig. \ref{fig:dist2}, the degree distributions with $\mu$=0.075 and 0.15 are not similar to ``long tail'' distributions but skew distributions with left skewness for $\mu$=0.075 and right skewness for $\mu$=0.15. This indicates that most vertices have a medium degree, and a few vertices have a small or large degree. However, in sub-figures \ref{fig:dist2} (c)-(f), the degree distributions are heterogenous featured with a ``long tail'' which obeys the Pareto Principle. The reason for that is the expression for the increase rate $\lambda(k)=ln(1+k)$ massively reduces the ``the rich gets richer'' effect compared to $\lambda(k)=k$. Additionally, the probability of a small degree also gets higher as $\mu$ increases, which indicates that degree distribution is more heterogenous as $\mu$ increases.

For Fig. \ref{fig:dist2}, there is an obvious phase change in the degree distributions. We further investigate this phenomenon and analyze the critical values. By simulation, we find that there are two thresholds for the degree decrease rate $\mu s$ which essentially change the distribution plots. In Fig. \ref{fig:fangda} (a), given $\mu=0.368$, the probability of $k=2$ denoted by $P_2$ is larger than $P_3$, while $P_2<P3$ when $\mu<0.366$. Therefore, there is a threshold of $\mu$ between 0.366 and 0.368 for $P_1=P_2$. With $\mu$ decreasing, in Fig. \ref{fig:fangda} (b), $P_2$ is larger than $P_1$ when $\mu>0.347$, while $P_1<P2$ when $\mu<0.346$. The threshold of $\mu$ for $P_1=P_2$ is between 0.346 and 0.347. The reason for the phase transition is theoretically related to Eq. \ref{eq:s}, where $\mu$ effects the value of $S(\mu)$ further effecting $P_k$. Thus, there are thresholds of $\mu$ for $P_1=P_2$ and $P_2=P_3$ which totally change the shape of the distribution plots.


For a better illustration, we present theoretical distributions under two forms of $\lambda(k)$ and with different values of $\mu$ as well as a Power-law distribution with the expression $2m^{2}k^{-3}$ for comparison under a logarithm coordinate in Fig. \ref{fig:logdist}.
In Fig. \ref{fig:logdist}, the distributions of both situations are not typical Power-law distributions though they have a ``long-tail'' character similar to Power-law distributions under ordinary coordinates. In detail, the red curve in Fig. \ref{fig:logdist} (b) representing the distribution with $\mu$=0.4 under $\lambda(k)=ln(1+k)$ is not monotonically decreasing, which corresponds to a skew distribution like distributions in sub-figures. \ref{fig:dist2} (a) and (b).
In Fig. \ref{fig:logdist}, in the $\lambda(k)=k$ cases, our degree distribution shows that more vertices have a relatively small degree, while fewer vertices have a quite small or quite large degree in comparison with the Power-law distribution, while distributions with $\lambda(k)=ln(1+k)$ manifest different forms (a skew distribution or a ``long-tail'' distribution) due to the value of $\mu$.

To measure the similarity of theoretical distributions and simulation results, we also calculate the Pearson correlation coefficient $\rho$, the Kullback-Leiber divergence denoted as KL, and the Jensen-Shannon divergence \cite{js} denoted as JS of them. $\rho$ is calculated by
\begin{equation}
\rho=\frac{E[(P_{ai}-\bar{P_{a}})(P_{bi}-\bar{P_{b}})]}{\sigma_{P_a}}{\sigma_{P_b}}
\end{equation}
where the range of $\rho$ is [-1,1].
\begin{figure}[!t]
\begin{minipage}[t]{0.48\linewidth}
\centering
\includegraphics[scale=0.29]{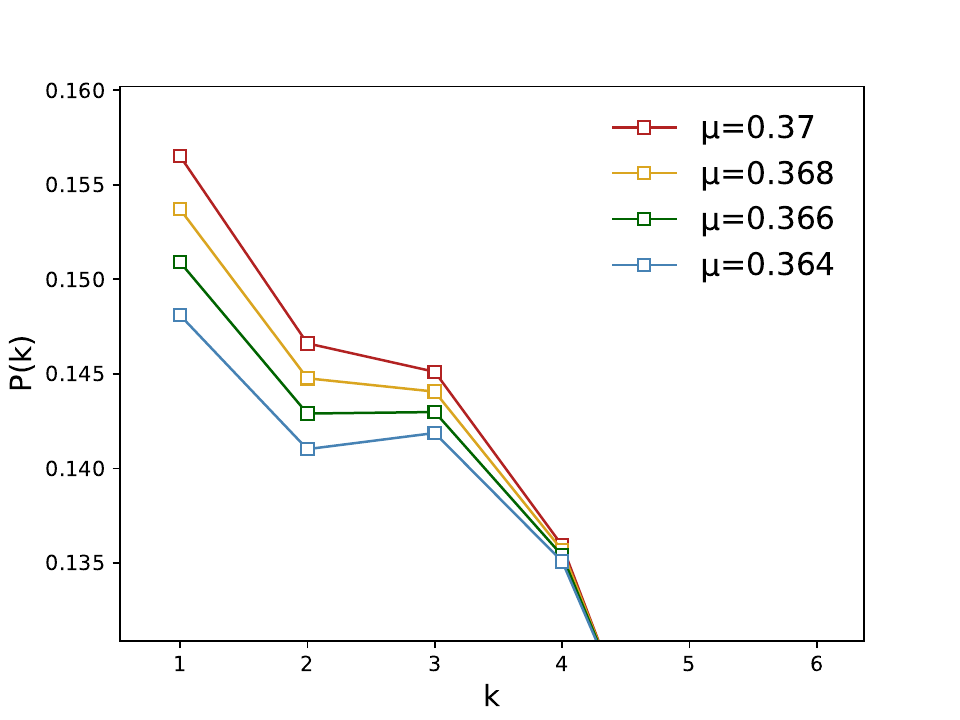}
\parbox{4cm}{\footnotesize\center (a) The threshold of $P_1$=$P_2$}
\end{minipage}
\hspace{0.15cm}
\begin{minipage}[t]{0.48\linewidth}
\centering
 \includegraphics[scale=0.29]{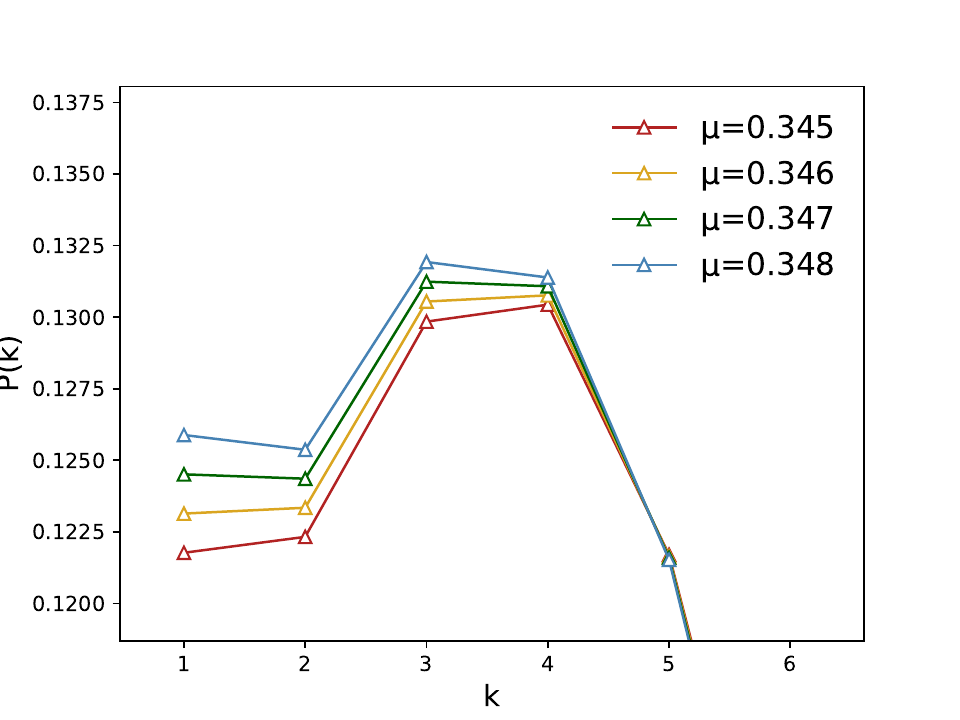}
\parbox{4cm}{\footnotesize \center(b) The threshold of $P_2$=$P_3$}
\end{minipage}
 \caption{\label{fig:fangda} Local map of the degree distributions under $\lambda(k)$=$ln(1+k)$ with different values of $\mu$. We magnify the local section of degree distributions with the form $\lambda=ln(1+k)$ to observe the phase transition in the distribution. (a) The value of $\mu$ is set as 0.364, 0.366, 0.368, 0.37 marked by squares with different colors. The threshold of $\mu$ for $P_1=P_2$ is between 0.366 and 0.368. (b) The value of $\mu$ is set as 0.345, 0.346, 0.347, 0.348 marked by triangles with different colors, and the threshold of $\mu$ for $P_2=P_3$ is between 0.346 and 0.347.}
\end{figure}

\begin{figure}[hpt]
\begin{minipage}[t]{0.48\linewidth}
\centering
\includegraphics[scale=0.29]{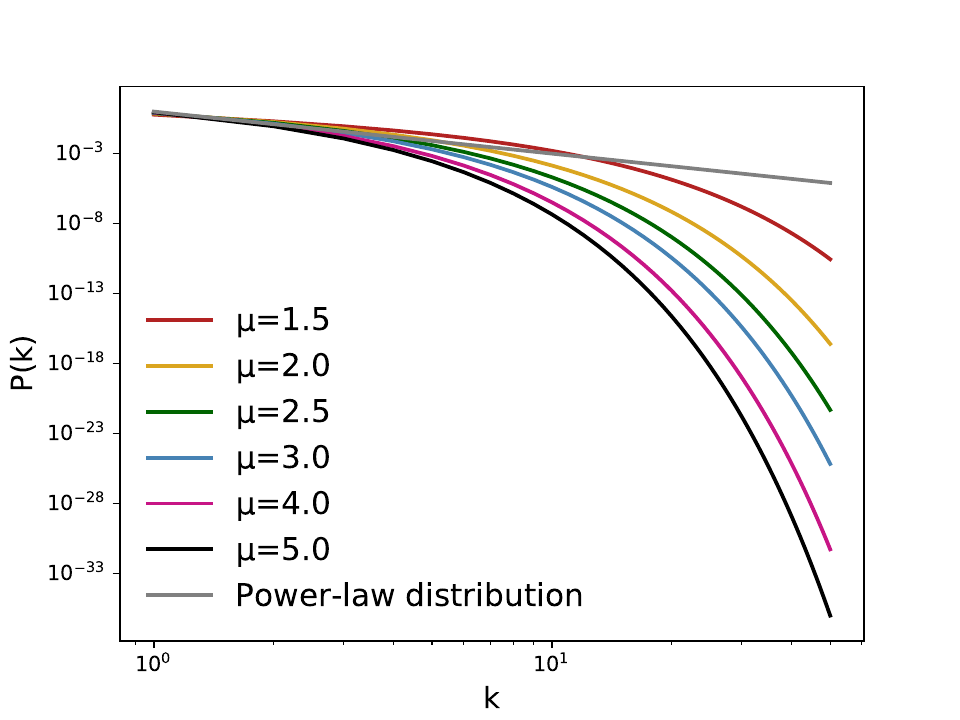}
\parbox{4cm}{\footnotesize\center (a) $\lambda(k)$=$k$}
\end{minipage}
\hspace{0.15cm}
\begin{minipage}[t]{0.48\linewidth}
\centering
 \includegraphics[scale=0.29]{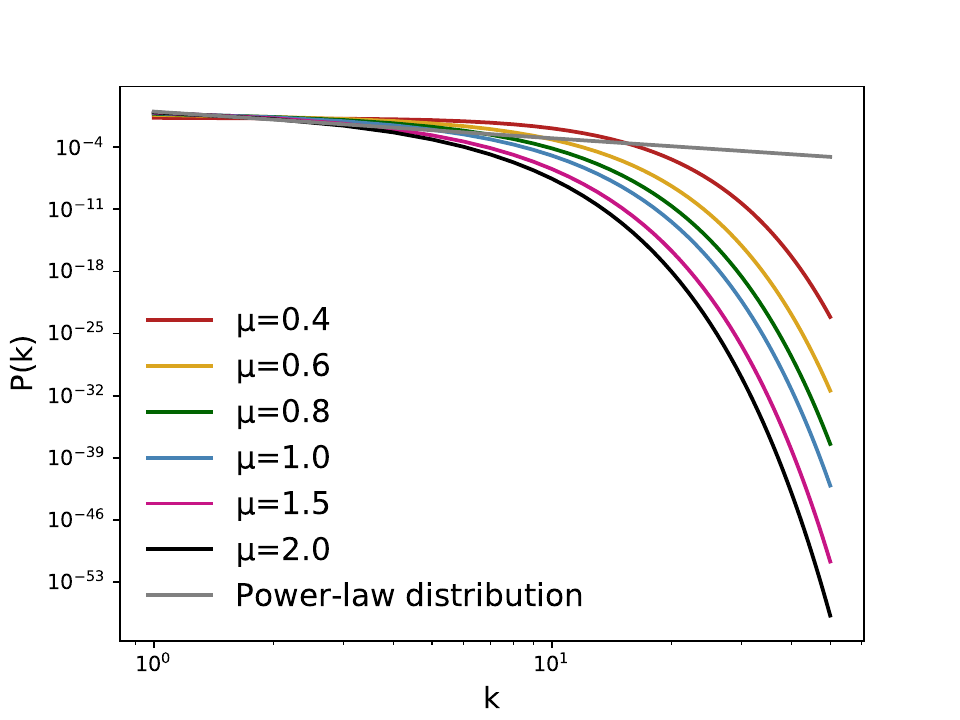}
\parbox{4cm}{\footnotesize \center(b) $\lambda(k)$=$ln(k+1)$}
\end{minipage}
 \caption{\label{fig:logdist} Degree distributions as a function of the degree decrease rate $\mu$ under two forms of the degree increase rate $\lambda(k)$:
 Degree distributions we present in the logarithmic coordinate are theoretical distributions we obtained according to our model. Distributions with different values of $\mu$ are marked by different colors, and the Power-law distribution is colored grey. According to the distribution plot in logarithmic coordinates, the theoretical degree distributions in our model are not Power-law distributions.}
\end{figure}
The Kullback-Leibler divergence is an asymmetric measure of the difference between two distributions, expressed as
\begin{equation}\label{eq:kl}
KL(P_a||P_b)=\sum P_alog\frac{P_a}{P_b}
\end{equation}
where the divergence is in the range of $(0,\infty)$.

The Jensen-Shannon divergence is based on the KL divergence but symmetrical, expressed as
\begin{equation}\label{eq:kl}
JS(P_a||P_b)=\frac{1}{2}KL(P_a||P_b)+\frac{1}{2}KL(P_b||P_a)
\end{equation}
which ranges between 0 and 1.

The results of these three measures under the two situations for $\lambda(k)$ with different values of $\mu$ are illustrated in Tab. \ref{tab:0}. We see that the Pearson correlation coefficients of theoretical and simulation distributions are all larger than 0.9 except for one value under the parameter of $\lambda(k)=k$ and $\mu=1.5$. The largest value of KL divergence and JS divergence is respectively 0.82 and 0.14. The small JS divergences in all situations suggest that the theoretical distributions agree with simulation distributions on the whole.

\subsection{Simulations of fitting real networks}
In this subsection, we apply our model to fit real-world networks from three different fields, a co-authorship network \cite{data1}, an email communication network \cite{data2}, and an autonomous system network \cite{data3}. For comparison, we also apply one of the latest evolving network models, MoncSid-N \cite{Moncsid} to fit these networks. Application analysis is performed by measuring the goodness of fitting via similarity measurement.


\subsubsection{Co-authorship network}
The distribution marked by blue circles illustrated in Fig. \ref{fig:e1} is the degree distribution of a co-authorship network. The practical network represents co-authorships in the area of network science with a total vertex umber of 1461 and a total edge number of 2742, where nodes represent authors and edges represent cooperation between authors. The network is undirected and unweighted and does not contain loops.

As is illustrated in Fig. \ref{fig:e1} marked by blue circles, the degree distribution of the co-authorship network is heterogenous
The fraction is below 0.1 when $k$=4, and it is very small and stays steady when $k>7$ featured with a ``long tail''. However, the degree distribution of the co-authorship network is not a typical Power-law distribution, as we can see in sub-figures (b) and (d). The degree distribution manifests as a curve going downward but not a straight line plot in the logarithmic coordinate.
Figs. \ref{fig:e1} (a) and (b) show the fitting results of $\lambda(k)=k$ where we set the parameter $\mu=1.28$ for the situation of $\lambda(k)=k$, and in Figs. \ref{fig:e1} (c) and (d), we set the parameter $\mu=0.58$ for $\lambda(k)=ln(1+k)$. In both sub-figures, most circles are around the theoretical lines, suggesting that the theoretical distribution is in accordance with the practical degree distribution. Under the logarithm coordinate in sub-figures \ref{fig:e1} (b) and (d), blue circles are not around the theoretical line when $k>10$. This is because the fraction of vertices with a degree larger than 10 is the same in the practical network owing to a small volume of data. Nevertheless, our theoretical distributions of both expressions for $ln(1+k)$ have a good fitting on the whole.


\begin{figure}[t]
\centering

\begin{minipage}[t]{0.5\linewidth}
\centering
\includegraphics[scale=0.28]{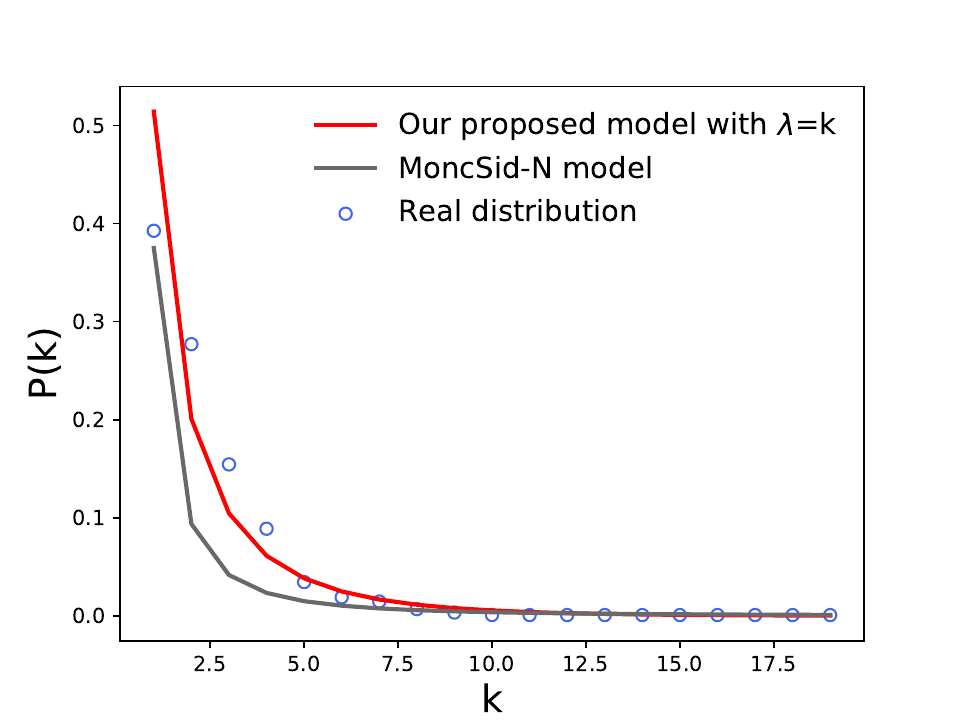}
\parbox{5cm}{\footnotesize \hspace{0.5cm}(a) $\lambda(k)$=$k$ under the log COOR}
\end{minipage}%
\begin{minipage}[t]{0.5\linewidth}
\centering
\includegraphics[scale=0.28]{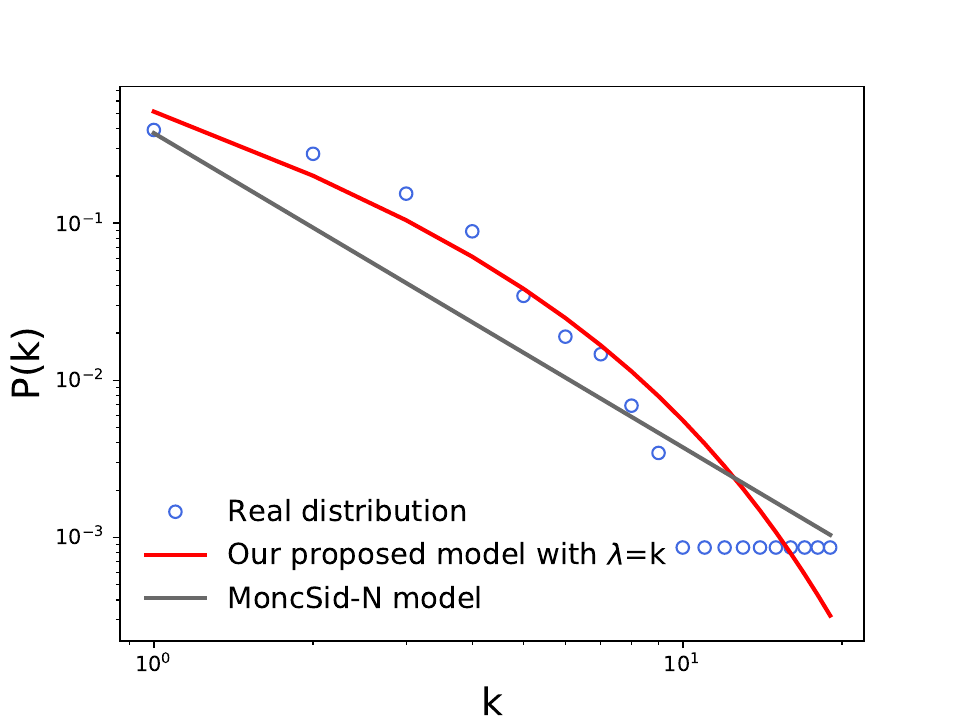}
\parbox{5cm}{\footnotesize \hspace{0.5cm}(b) $\lambda(k)$=$k$ under the normal \par \qquad \qquad \qquad \quad COOR}
\end{minipage}%

\begin{minipage}[t]{0.5\linewidth}
\centering
\includegraphics[scale=0.28]{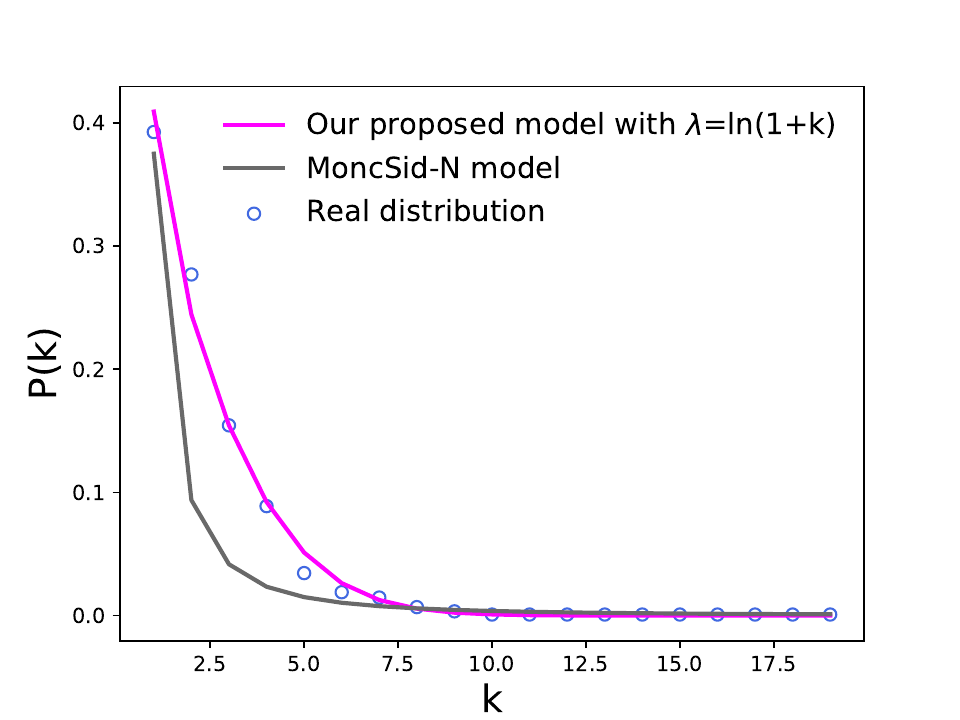}
\parbox{5cm}{\footnotesize \hspace{0.5cm}(c) $\lambda(k)$=$ln(1+k)$ under the log \par \qquad \qquad \qquad \quad COOR}
\end{minipage}%
\begin{minipage}[t]{0.5\linewidth}
\centering
\includegraphics[scale=0.28]{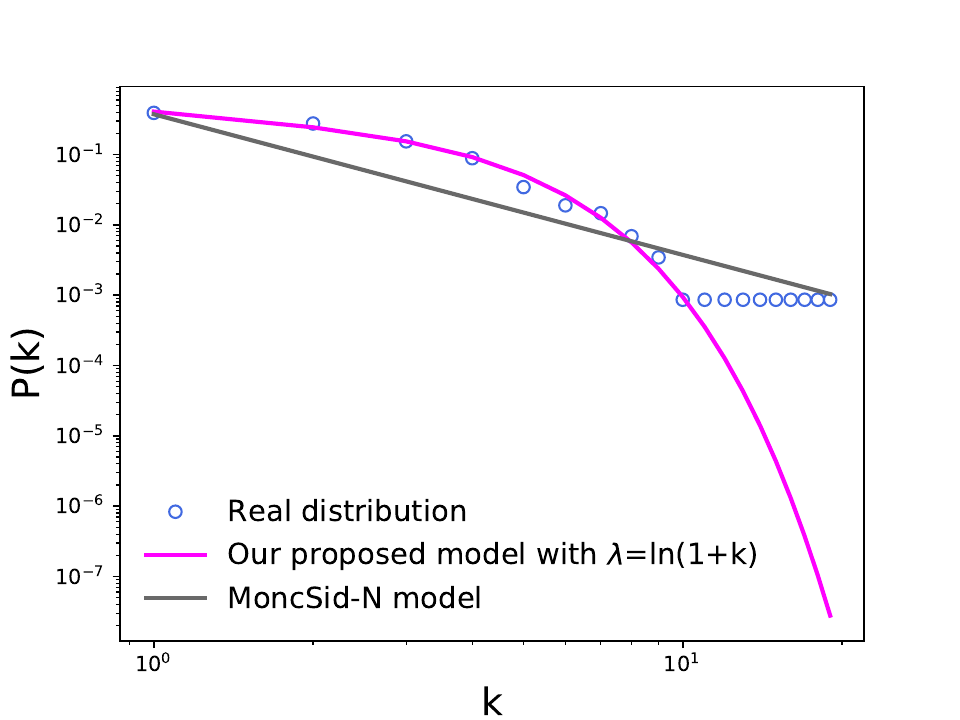}
\parbox{5cm}{\footnotesize \hspace{0.5cm}(d) $\lambda(k)$=$ln(1+k)$ under the \par \qquad \qquad \qquad normal COOR}
\end{minipage}%

\centering
\caption{\label{fig:e1} Comparison of the theoretical and real degree distributions of the co-authorship network: Two situations of the increase rate $\lambda(k)=k$ and $\lambda(k)=ln(1+k)$ of our degree queueing system model are respectively applied to the co-authorship in the field of network science. For each form of $\lambda(k)$, we use different colors to present the theoretical degree distribution and the real degree distribution under the normal coordinate and the logarithmic coordinate, turquoise for
$\lambda(k)=k$ and magenta for $\lambda(k)=ln(1+k)$. The MoncSid-N model is also applied to fit the real network for comparison, marked by the grey line.}
\end{figure}

\begin{figure}[htbp]
\centering

\begin{minipage}[t]{0.5\linewidth}
\centering
\includegraphics[scale=0.28]{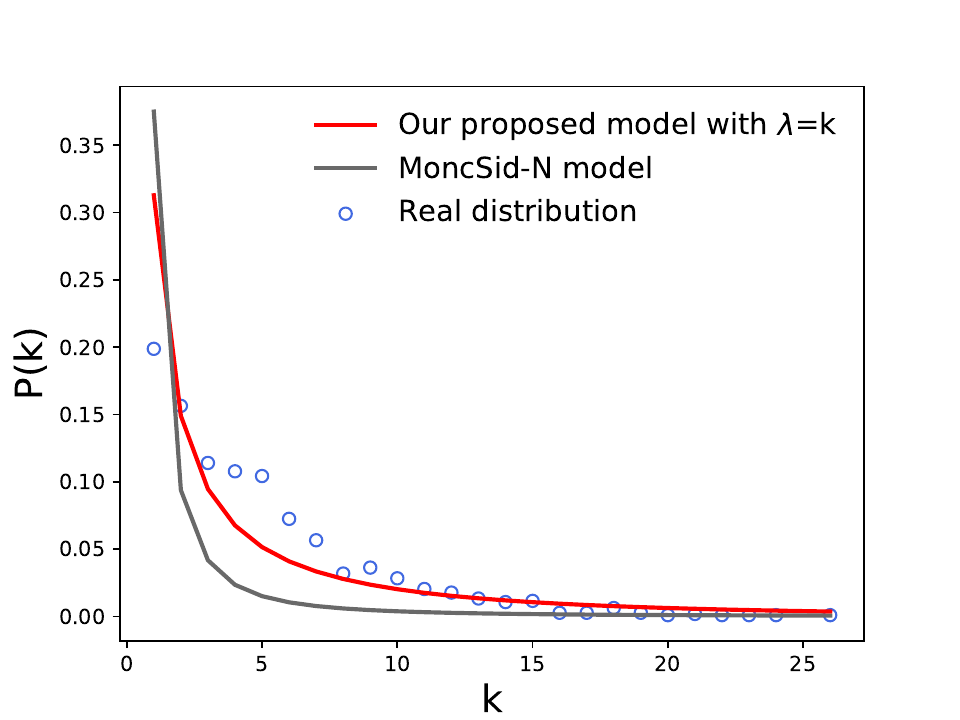}
\parbox{5cm}{\footnotesize \hspace{0.5cm}(a) $\lambda(k)$=$k$ under the log COOR}
\end{minipage}%
\begin{minipage}[t]{0.5\linewidth}
\centering
\includegraphics[scale=0.28]{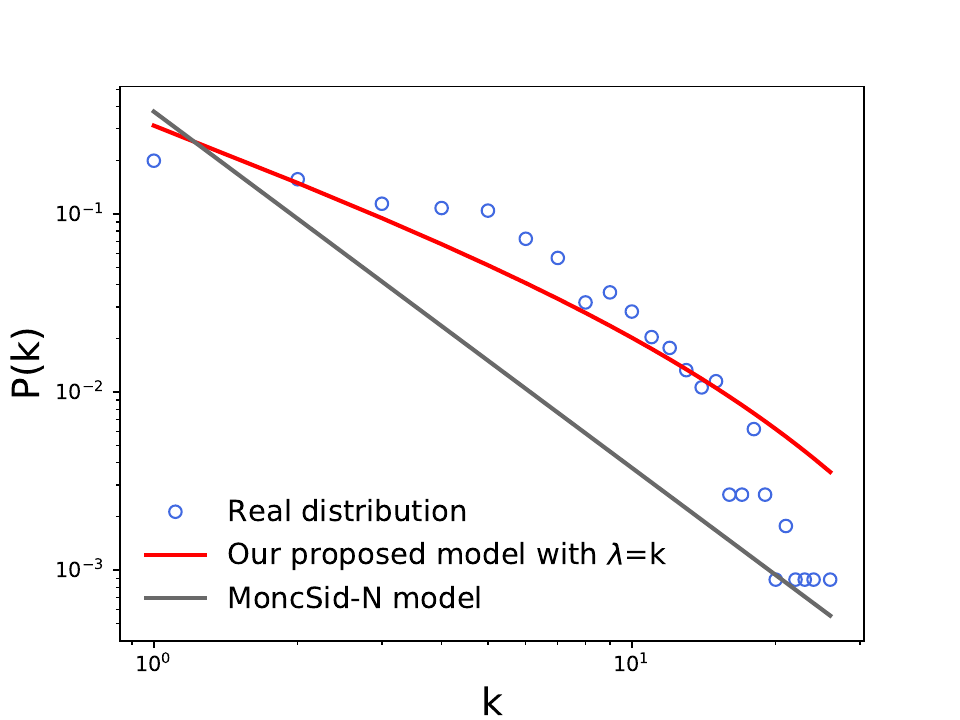}
\parbox{5cm}{\footnotesize \hspace{0.5cm}(b) $\lambda(k)$=$k$ under the normal \par \qquad \qquad \qquad \quad COOR}
\end{minipage}%

\begin{minipage}[t]{0.5\linewidth}
\centering
\includegraphics[scale=0.28]{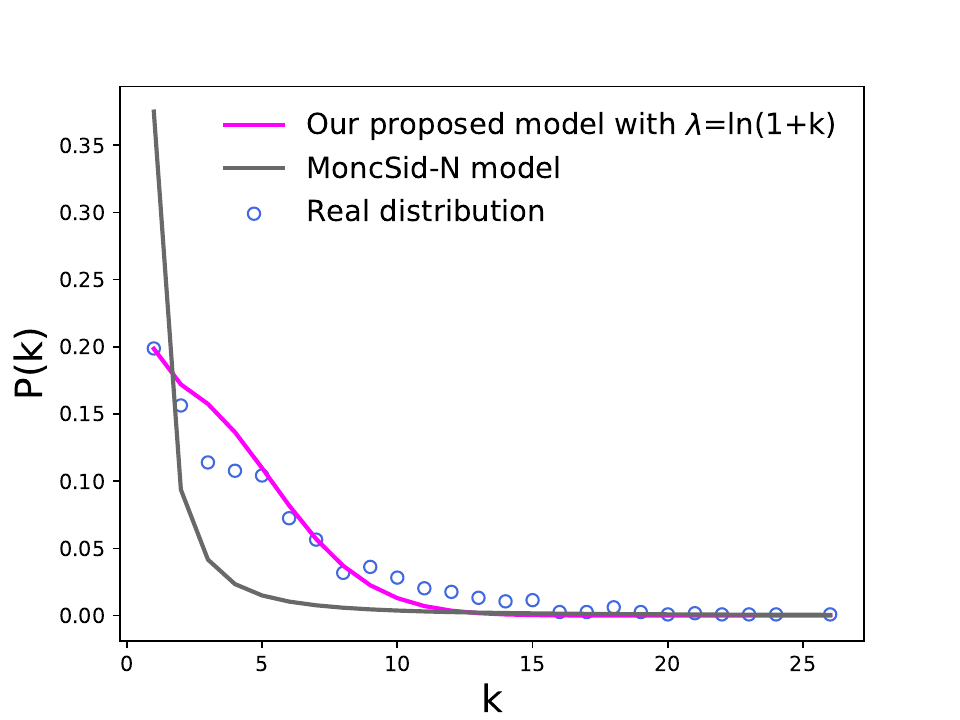}
\parbox{5cm}{\footnotesize \hspace{0.5cm}(c) $\lambda(k)$=$ln(1+k)$ under the log \par \qquad \qquad \qquad \quad COOR}
\end{minipage}%
\begin{minipage}[t]{0.5\linewidth}
\centering
\includegraphics[scale=0.28]{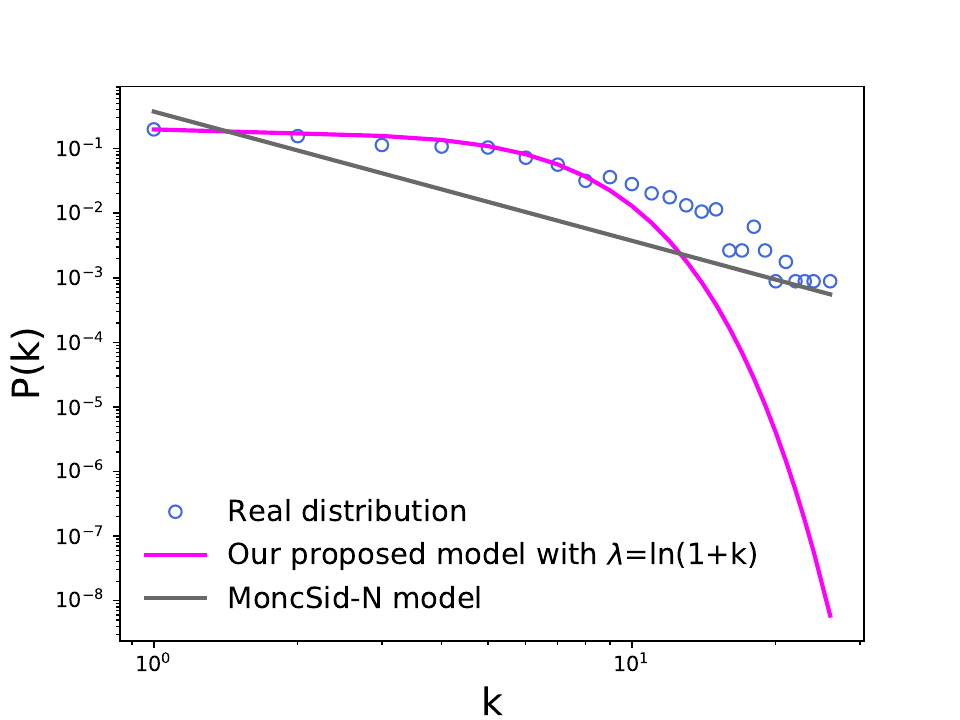}
\parbox{5cm}{\footnotesize \hspace{0.5cm}(d) $\lambda(k)$=$ln(1+k)$ under the \par \qquad \qquad \qquad normal COOR}
\end{minipage}%

\centering
\caption{\label{fig:e2} Comparison of the theoretical and real degree distributions of the email communication network: Two situations of the increase rate $\lambda(k)=k$ and $\lambda(k)=ln(1+k)$ of our degree queueing system model are respectively applied to the email communication at the University Rovira i Virgili. For each form of $\lambda(k)$, two different colors are selected to present the theoretical degree distribution and the real degree distribution under the normal coordinate and the logarithmic coordinate, turquoise for $\lambda(k)=k$ and magenta for $\lambda(k)=ln(1+k)$. The MoncSid-N model is also applied to fit the real network for comparison, marked by the grey line.}
\end{figure}

\subsubsection{Email communication network}
The practical degree distribution of an email communication network is illustrated in Fig. \ref{fig:e2}. This network represents co-authorships in the area of network science with a total vertex umber of 1133 and a total edge number of 5451, where nodes indicate users and edges represent that at least one email was sent between two users. The network is undirected and unweighted, as well as excluding loops.

In Figs. \ref{fig:e2} (a) and (c), the practical distribution of the email communication network is also heterogenous.
The degree distribution of this practical network is also not a Power-law distribution but a downward curve as we can see in the logarithm coordinate shown in Fig. \ref{fig:e2} (b) and (d), yet it is featured with a ``long tail'' shown in (a) and (c).
We set the parameter $\mu=1.05$ for the situation of $\lambda(k)=k$ presented in Figs. \ref{fig:e2} (a) and (b), and set the parameter $\mu=0.40$ for $\lambda(k)=ln(1+k)$ illustrated in Figs. \ref{fig:e2} (c) and (d). From Figs. \ref{fig:e2} (a) and (c), we can see that the probability of $k$=1 is 0.3 according to the theoretical distribution of $\lambda(k)=k$, higher than the real degree distribution with a probability $P(k)$=0.2, however, they are close on the whole. Comparing sub-figures (a) and (c), the theoretical line with $\lambda=k$ is closer to real data circles than the line with $\lambda=ln(1+k)$, which suggests that the theoretical distribution with $\lambda(k)=ln(1+k)$ has better goodness of fitting than that with $\lambda(k)=ln(1+k)$ from the normal coordinate. From the plots under the logarithm coordinate shown in Fig. \ref{fig:e3} (b), the theoretical line is like a straight line with a small curvature, indicating that the theoretical distribution is close to a Power-law distribution. The circle plot representing the practical degree distribution has a larger curvature compared to the theoretical line, which leads to the tail of the line does not agree with the circle plot. The situation in Fig. \ref{fig:e2} (d) is different, the curvature of the theoretical line is larger than that of the practical distribution.

To view the plots as a whole, two theoretical distributions fit the real degree distribution well. 


\begin{figure}[htbp]
\centering

\begin{minipage}[t]{0.5\linewidth}
\centering
\includegraphics[scale=0.28]{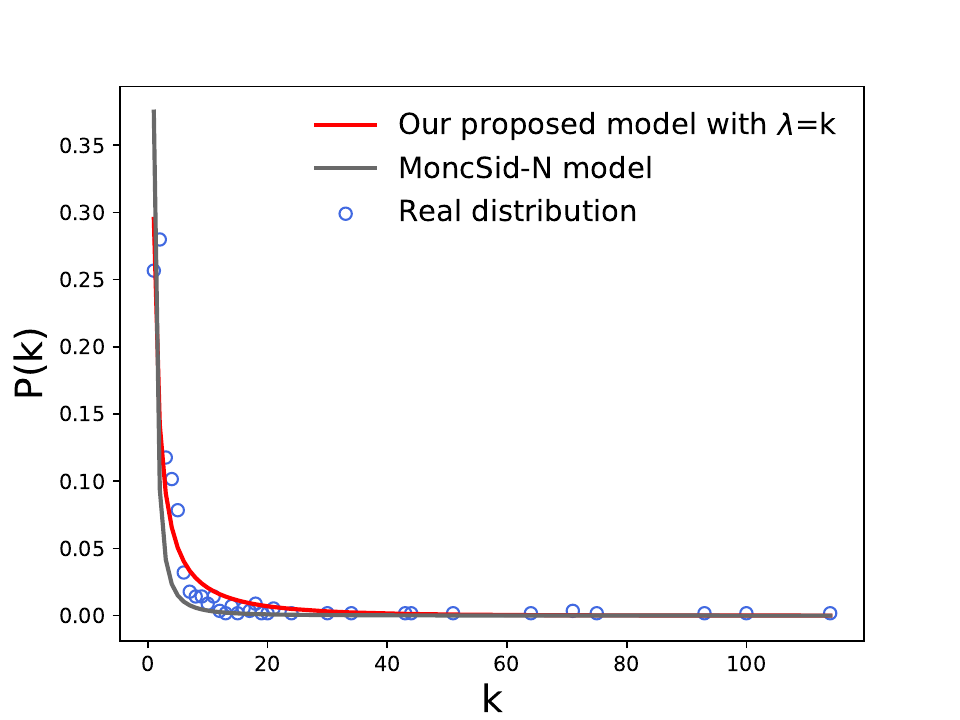}
\parbox{5cm}{\footnotesize \hspace{0.5cm}(a) $\lambda(k)$=$k$ under the log COOR}
\end{minipage}%
\begin{minipage}[t]{0.5\linewidth}
\centering
\includegraphics[scale=0.28]{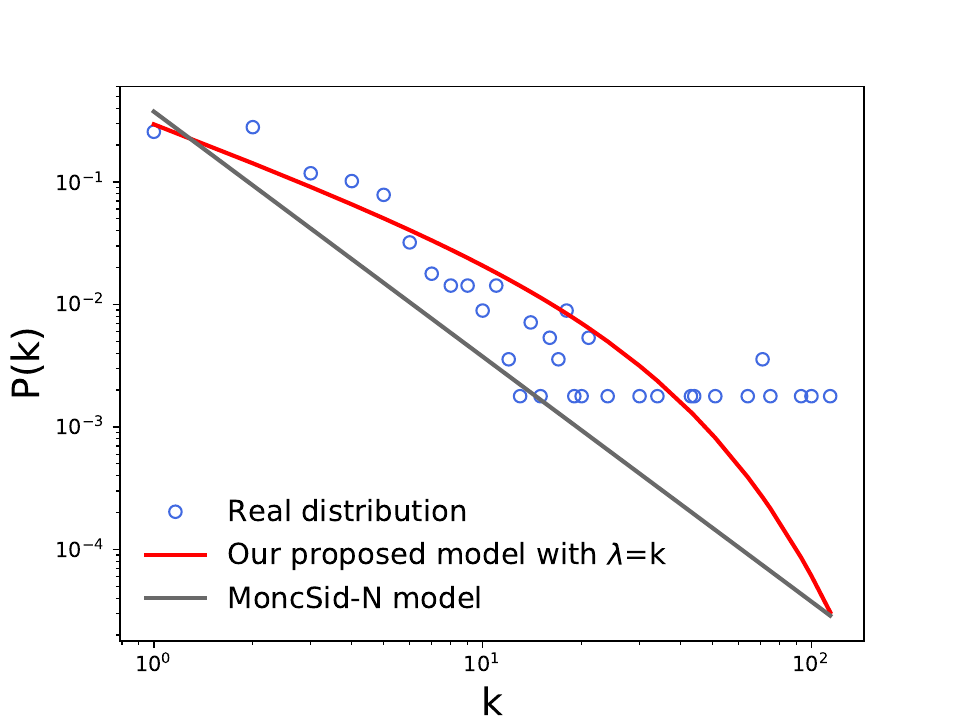}
\parbox{5cm}{\footnotesize \hspace{0.5cm}(b) $\lambda(k)$=$k$ under the normal \par \qquad \qquad \qquad \quad COOR}
\end{minipage}%

\begin{minipage}[t]{0.5\linewidth}
\centering
\includegraphics[scale=0.28]{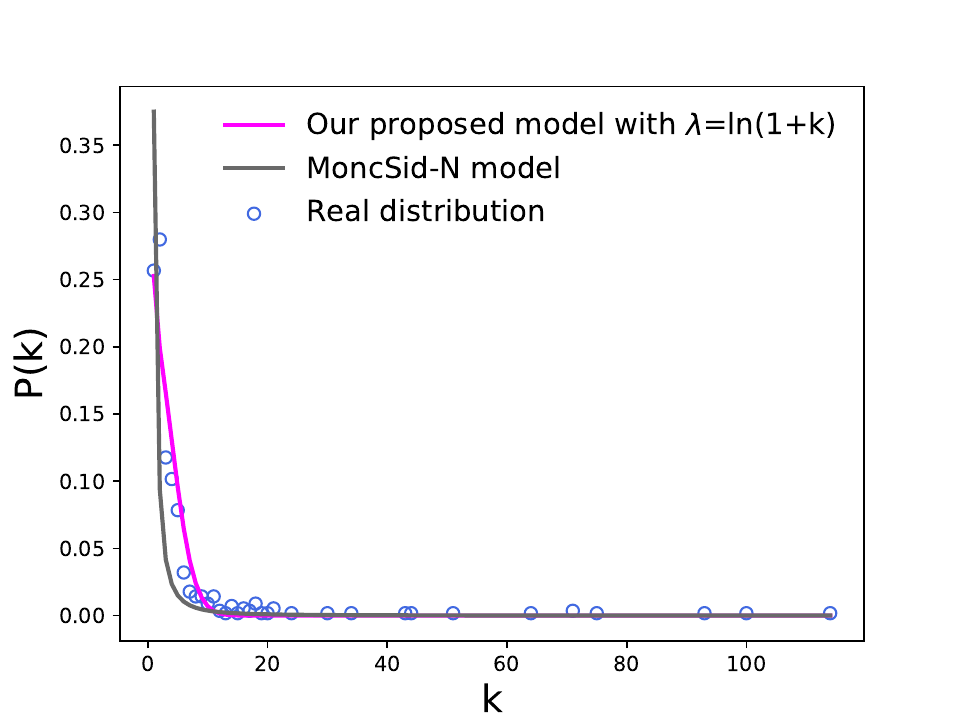}
\parbox{5cm}{\footnotesize \hspace{0.5cm}(c) $\lambda(k)$=$ln(1+k)$ under the log \par \qquad \qquad \qquad \quad COOR}
\end{minipage}%
\begin{minipage}[t]{0.5\linewidth}
\centering
\includegraphics[scale=0.28]{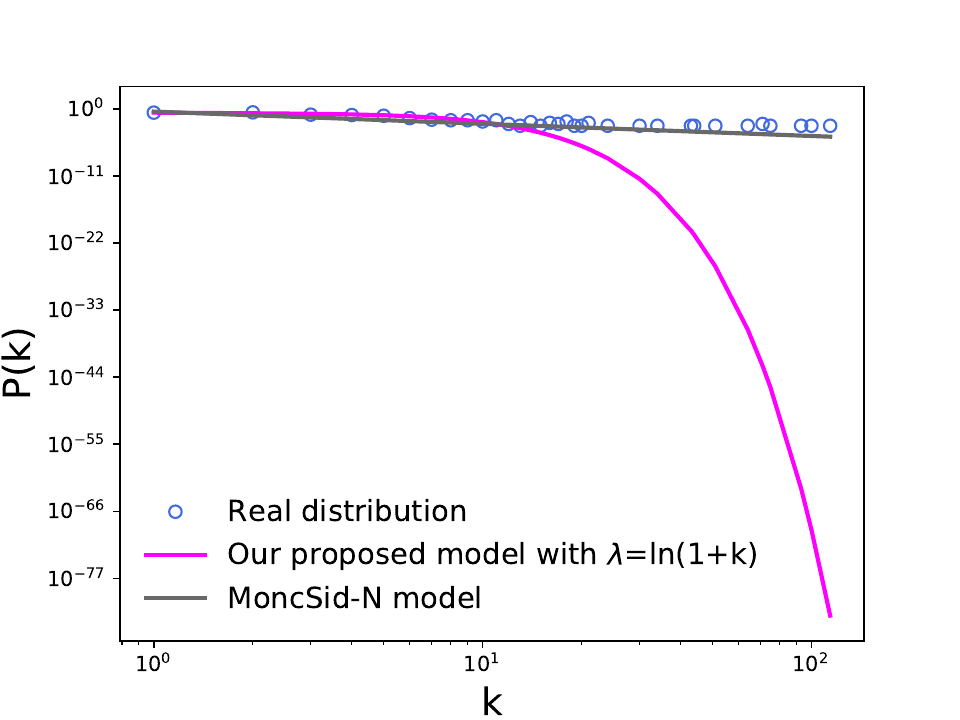}
\parbox{5cm}{\footnotesize \hspace{0.5cm}(d) $\lambda(k)$=$ln(1+k)$ under the \par \qquad \qquad \qquad normal COOR}
\end{minipage}%

\centering
\caption{\label{fig:e3} Comparison of the theoretical and real degree distributions of the autonomous system (AS) network: Two forms of the increase rate $\lambda(k)=k$ and $\lambda(k)=ln(1+k)$ of our degree queueing system model are respectively applied to the AS network. For each form of $\lambda(k)$, two different colors are selected to present the theoretical degree distribution and the real degree distribution under the normal coordinate and the logarithmic coordinate, turquoise for $\lambda(k)=k$ and magenta for $\lambda(k)=ln(1+k)$. The MoncSid-N model is also applied to fit the real network for comparison, marked by the grey line.}
\end{figure}

\begin{table*}[htbp]
  \centering
  \caption{Results and parameters of degree distribution fitting of three practical networks}
    \begin{tabular}{cccc|ccc|ccc}
    \hline
    \hline
          & \multicolumn{3}{c|}{Co-authorships} & \multicolumn{3}{c|}{Email communication} & \multicolumn{3}{c}{Autonomous systems} \\
\cline{2-10}       & $\lambda(k)=k$   & $\lambda(k)=ln(1+k)$ & MoncSid-N & $\lambda(k)=k$   & $\lambda(k)=ln(1+k)$ & MoncSid-N & $\lambda(k)=k$ & $\lambda(k)=ln(1+k)$ & MoncSid-N \\
    $\mu$     & 1.28  & 0.58 &  & 1.05  & 0.40 & & 1.04  & 0.44 & \\
    $\rho$     & 0.9581 & 0.9960 & 0.8992 & 0.9055 & 0.9808 & 0.7602 & 0.9162 & 0.9580 & 0.7889 \\
    KL    & 0.0549 & 0.0421 & 0.1672 & 0.0990 & 0.1752 & 0.5212 & 0.1457 & 2.0401 & 0.3274 \\
    JS    & 0.0139 & 0.0040 & 0.0419 & 0.0260 & 0.0219 & 0.1265 & 0.0354 & 0.0341 & 0.0813 \\
    \hline
    \hline
    \end{tabular}%
  \label{tab:real}%
\end{table*}%
\subsubsection{Autonomous systems traffic-flow network}
The practical degree distribution of an autonomous system (AS) traffic-flow network is presented in Fig. \ref{fig:e3} marked by circles. AS is sub-graphs of routers comprising the Internet, where nodes represent routers and edges are traffic flows between two routers. The whole dataset contains 733 networks over 785 days from November 8th 1997 to January 2nd 2000, and AS networks include both the addition and deletion of the nodes and edges over time. Hereby, we choose a snapshot of the evolving As network.

As is illustrated in Fig. \ref{fig:e3}, the degree distribution of the AS network is also a ``long-tail'' distribution, where
the fraction of vertices with degree $k$=1 is over 0.25, and the fractions of $k$>5 are all below 0.05.
We set $\mu$=1.04 for the theoretical distribution with $\lambda(k)=k$, and $\mu$=0.44 for that with $\lambda(k)=ln(1+k)$, respectively demonstrated in Figs. \ref{fig:e3} (a) and (c) in normal coordinates, (b) and (d) in logarithmic coordinates.
For $\lambda(k)=k$, most circles are around the theoretical line under normal coordinates as is illustrated in Fig. \ref{fig:e3} (a), suggesting that the theoretical distribution with $\lambda(k)=k$ fits the real degree distribution well. While under the logarithm coordinate in Fig. \ref{fig:e3} (b), the trend of tail circles does not correspond to the theoretical line. This is because of the non-uniform scale of axes of logarithm coordinates which magnifies the difference between the real distribution and the theoretical one when $k$ is getting large. For $\lambda(k)=ln(1+k)$, under the normal coordinate, most circles are also around the theoretical line except for $k$=2 whose theoretical probability is 0.25 approximately, less than the practical fraction which is about 0.28. However, under the logarithm coordinate, the tail of the theoretical line with $\lambda(k)=ln(1+k)$ is too much below the blue circles Fig.\ref{fig:e3} (d), which shows a poor fitting result. In general, there is superb goodness of fit of theoretical distribution with $\lambda(k)=k$ compared to $\lambda(k)=ln(1+k)$ for the real degree distributions of AS network.

We measure the goodness of fitting for practical degree distributions based on three measures introduced in sub-section \ref{sec:IIIB}, as well as the values of parameter $\mu$ are displayed in Tab. \ref{tab:real}, where the values of $\rho$ are all larger than 0.90, and the JS divergence are all less than or equal to 0.0354. For the KL divergence, except for the value of the AS network with $\lambda(k)=ln(1+k)$, the values of other fitting results are all less than or equal to 0.1752. This shows good fitting results for three real networks via applying our theoretical distributions.


For comparison, we also utilize other evolving network models, MoncSid-N \cite{Moncsid} to fit real networks by comparing their degree distributions. The MoncSid-N model is marked by grey lines in Fig. \ref{fig:e1}, \ref{fig:e2} and \ref{fig:e3}, where we can see that under the logarithm coordinate, the degree distributions of MoncSid-N are Power-law distributions, manifesting a straight line. In detail, for the co-authorship network, the degree distributions of our model $\lambda(k)=k$ and $\lambda(k)=ln(1+k)$ are both closer to the circle plot of the real distribution compared to the MoncSid-N model, shown in Fig. \ref{fig:e1}. For better validation, in Tab. \ref{tab:real}, the Person correlation of our theoretical distributions is larger than that of the MoncSid-N model, while the KL and JS divergence of our theoretical distribution are smaller compared to that of the MoncSid-N model. A similar situation can be found for the email communication network in Fig. \ref{fig:e2}. For the autonomous system network, our proposed model with $\lambda(k)=k$ is more appropriate for describing its real degree distribution, compared to the MoncSid model and our model with $\lambda(k)=ln(1+k)$, presented in Fig. \ref{fig:e3}. Correspondingly, in Tab. \ref{tab:real}, in terms of fitting the autonomous system network, our model with $\lambda(k)=k$ shows its advantages according to the three metrics.

\section{\label{sec:V}Conclusions and Outlook}
In this paper, we propose an evolving network from the perspective of the vertex degree variation. We let the degree increase rate in the degree queueing system be a function positively correlated with the degree, which promises the preferential attachment mechanism in many real-world networks. The degree distribution of our model is analyzed by a Markov process under two situations for the degree increase rate $\lambda(k)$ that are $\lambda(k)=k$ and $\lambda(k)=ln(1+k)$. The new distributions we obtained are heterogenous and featured with a ``long tail'', but not typical Power-law distributions as we see under the logarithm coordinate. It can well describe degree distributions of some real networks in different areas.

There are still some issues worth further investigating. There are various other expressions for the degree increase rate $\lambda(k)$ as a function of the degree $k$ while we only discuss two specific situations. A general paradigm is expected to be formed for analyses of degree distributions. Besides, our network model can be hopefully applied to different scenarios in various fields. Our evolving networks can be combined with representation learning models, such as dynamic graph neural networks (GNN), and used to study the dynamic properties of Peer-to-Peer (P2P) networks. We will use the proposed evolving network model and mechanism to construct dynamic P2P networks with node additions and departures for further study. We will focus on these issues and applications in the near future.

\bibliographystyle{IEEETran}

\begin{thebibliography}{30}
\bibitem{WS}
Watts D J, Strogatz S H. Collective dynamics of "small-world" networks. Nature, 1998, 393(6684): 440-442.

\bibitem{SF}
Barab\'{a}si A L, Albert R. Emergence of scaling in random networks. Science, 1999, 286(5439): 509-512.


\bibitem{random}
Zhang X, He Z, Rayman-Bacchus L. Random birth-and-death networks. Journal of Statistical Physics, 2016, 162(4): 842-854.

\bibitem{weight}
Pi X, Tang L, Chen X. A directed weighted scale-free network model with an adaptive evolution mechanism. Physica A: Statistical Mechanics and its Applications, 2021, 572: 125897.


\bibitem{bd}
Sankararaman A, Baccelli F. Spatial birth\mbox{-}death wireless networks. IEEE Transactions on Information Theory, 2017, 63(6): 3964-3982.

\bibitem{tree}
Palacios J, Quiroz D. Birth and death chains on finite trees: computing their stationary distribution and hitting times. Methodology and Computing in Applied Probability, 2016, 18(2): 487-498.

\bibitem{feng1}
Feng M, Deng L, Kurths J. Evolving networks based on birth and death process regarding the scale stationarity. Chaos: An Interdisciplinary Journal of Nonlinear Science, 2018, 28(8): 083118.

\bibitem{feng2}
Feng M, Li Y, Chen F, Kurths J. Heritable Deleting Strategies for Birth and Death Evolving Networks From a Queueing System Perspective. IEEE Transactions on Systems, Man, and Cybernetics: Systems, 2022, 52(10): 6662-6673.


\bibitem{data}
Zhou D, Zheng L, Han J, He, J. A data-driven graph generative model for temporal interaction networks. In Proceedings of the 26th ACM SIGKDD International Conference on Knowledge Discovery \& Data Mining 2020, 401-411.

\bibitem{learning}
Bian R, Koh Y S, Dobbie G, Divoli A. Network embedding and change modeling in dynamic heterogeneous networks. In Proceedings of the 42nd International ACM SIGIR Conference on Research and Development in Information Retrieval, 2019, 861-864.

\bibitem{hoc}
Liu S, Zhang D, Liu X, Zhang T, Gao J, Cui Y. Dynamic analysis for the average shortest path length of mobile ad hoc networks under random failure scenarios. IEEE Access, 2019, 7: 21343-21358.

\bibitem{tem}
Holme P. Modern temporal network theory: A colloquium. The European Physical Journal B, 2015, 88(9): 234.

\bibitem{tem_based}
Hiraoka T, Masuda N, Li A, Jo H H. Modeling temporal networks with bursty activity patterns of nodes and links. Physical Review Research, 2020, 2(2), 023073.

\bibitem{ad}
Raimundo R L, Guimaraes Jr P R, Evans D M. Adaptive networks for restoration ecology. Trends in Ecology \& Evolution, 2018, 33(9): 664-675.


\bibitem{sad}
Petri G, Barrat A. Simplicial activity driven model. Physical review letters, 2018, 121(22): 228301.



\bibitem{dd2}
Sudbrack V, Brunnet L G, de Almeida R M C, Ferreira R M, Gamermann, D. Master equation for the degree distribution of a Duplication and Divergence network. Physica A: Statistical Mechanics and its Applications, 2018, 509: 588-598.

\bibitem{feng3}
Feng M, Qu H, Yi Z, Kurths J. Subnormal distribution derived from evolving networks with variable elements. IEEE Transactions on Cybernetics, 2017, 48(9): 2556-2568.

\bibitem{zhang}
Zhang X, He Z, He Z, Rayman-Bacchus L. SPR-based Markov chain method for degree distributions of evolving networks. Physica A: Statistical Mechanics and its Applications, 2012, 391(11): 3350-3358.

\bibitem{Cihan}
Cihan O, Akar M. Fastest mixing reversible Markov chains on graphs with degree proportional stationary distributions. IEEE Transactions on Automatic Control, 2014, 60(1): 227-232.

\bibitem{mar_exa}
Zhang X J, Yang H L. Degree distribution of random birth-and-death network with network size decline. Chinese Physics B, 2016, 25(6): 060202.

\bibitem{between}
Kourtellis N, Morales G D F, Bonchi F. Scalable online betweenness centrality in evolving graphs. IEEE Transactions on Knowledge and Data Engineering, 2015, 27(9): 2494-2506.

\bibitem{cc}
Yin H, Benson A R, Leskovec J. Higher-order clustering in networks. Physical Review E, 2018, 97(5): 052306.

\bibitem{assor}
Zhou B, Lu X, Holme P. Universal evolution patterns of degree assortativity in social networks. Social Networks, 2020, 63, 47-55.

\bibitem{js}
Nielsen F. On a variational definition for the jensen-shannon symmetrization of distances based on the information radius. Entropy, 2021, 23(4): 464.

\bibitem{nonhq}
Yue D, Zhao G, Yue W. Analysis of a multi-server queueing-inventory system with non-homogeneous poisson arrivals. Proceedings of the 11th International Conference on Queueing Theory and Network Applications, 2016: 1-5.

\bibitem{data1}
Network science network dataset, KONECT, 2018. [Online]. Available: http://konect.cc/networks/dimacs10-netscience.


\bibitem{data2}
U. Rovira i Virgili network dataset, KONECT, 2017. [Online]. Available: http://konect.cc/networks/arenas-email.


\bibitem{data3}
Leskovec J, Kleinberg J, Faloutsos C. Graphs over time: densification laws, shrinking diameters and possible explanations. Proceedings of the eleventh ACM SIGKDD international conference on Knowledge discovery in data mining, 2005, 177-187.

\bibitem{Moncsid}
Chen H, Chen B, Ai C, Zhu M, Qiu X. The evolving network model with community size and distance preferences. Physica A: Statistical Mechanics and its Applications, 2022, 596: 127112.
\end{thebibliography}
\begin{IEEEbiography}[{\includegraphics[width=1in,clip,keepaspectratio]{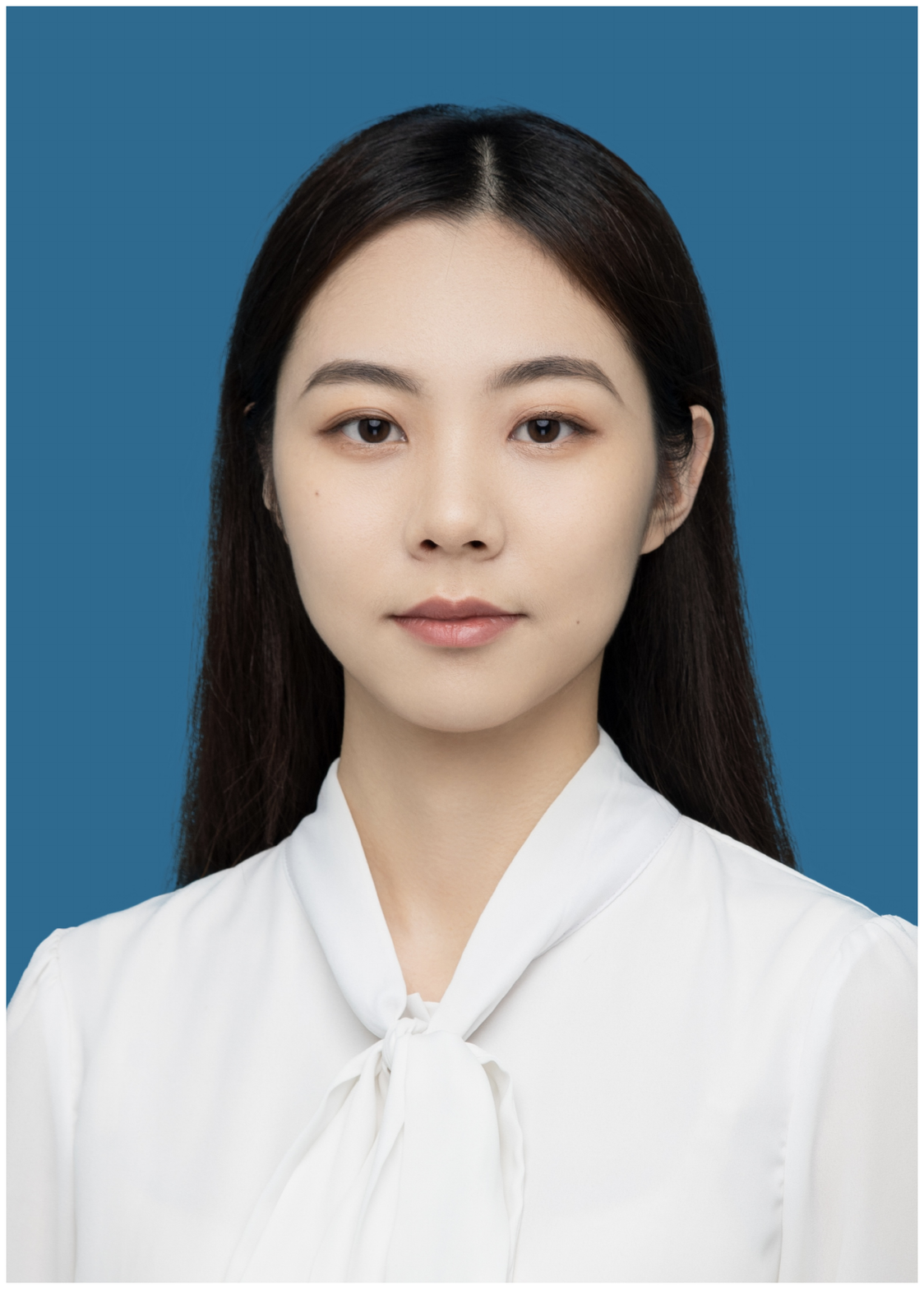}}]{Yuhan Li}
is pursuing the B.E. degree in data science and big data technology from the College of Artificial Intelligence, Southwest University, Chongqing, China. Her research interests include complex networks, contagious processes, epidemic modeling, stochastic process, queueing theory and their applications.
\end{IEEEbiography}

\begin{IEEEbiography}[{\includegraphics[width=1in,clip,keepaspectratio]{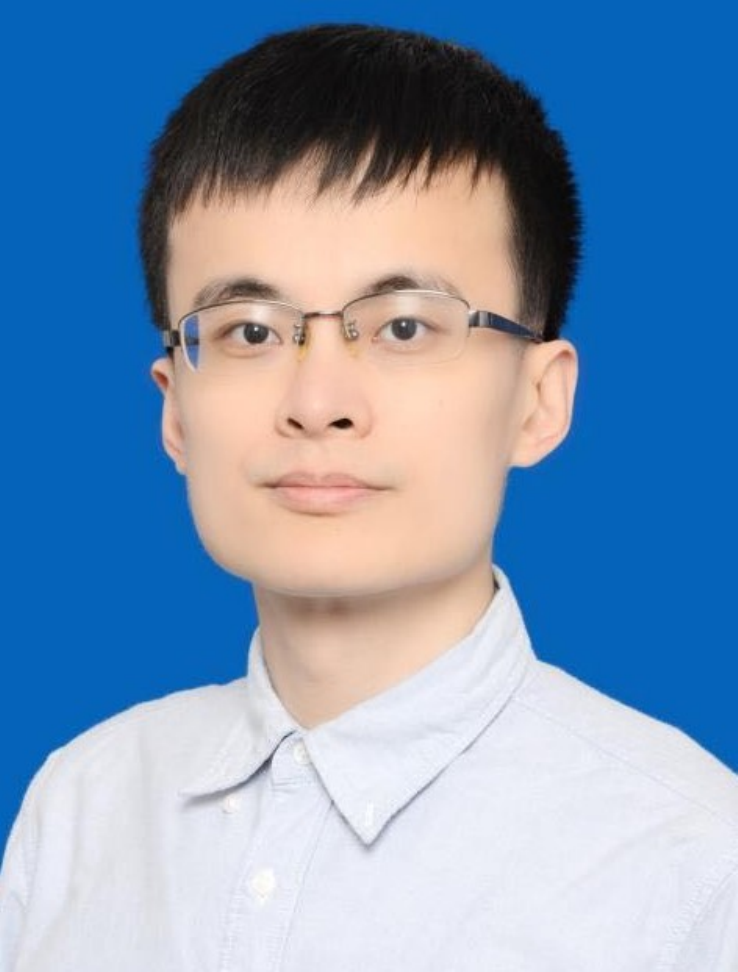}}]{Minyu Feng}
received the B.S. degree in mathematics from the University of Electronic Science and Technology of China in 2010; the Ph.D. degree in computer science from the University of Electronic Science and Technology of China in 2018. From 2016 to 2017, he was a visiting scholar with the Potsdam Institute for Climate Impact Research, Potsdam, Germany, and Humboldt University of Berlin, Berlin, Germany. Since 2019, he has been an associate professor in the College of Artificial Intelligence, Southwest University, Chongqing, China. He is an academic member of IEEE, CCF and CAA. He currently serves as a guest editor in Entropy and Frontiers in Physics. His research interests include stochastic processes, complex systems, evolutionary games and social computing.
\end{IEEEbiography}

\begin{IEEEbiography}[{\includegraphics[width=1in,clip,keepaspectratio]{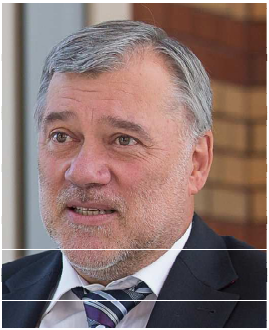}}]{J\"{u}rgen Kurths}
received the B.S. degree in mathematics from the University of Rostock; the Ph.D. degree from the Academy of Sciences of the German Democratic Republic in 1983; the Honorary degree from N.I. Lobachevsky State University of Nizhny Novgorod in 2008; and the Honorary degree from Saratow State University in 2012.

From 1994 to 2008, he was a full Professor
with the University of Potsdam, Potsdam, Germany.
Since 2008, he has been a Professor of nonlinear
dynamics with Humboldt University, and the Chair
of the Research Domain Complexity Science with the Potsdam Institute
for Climate Impact Research, Germany. He is the author of
more than 700 papers, which are cited more than 60.000 times (H-index: 111).
His main research interests include synchronization, complex networks, time
series analysis, and their applications.

Dr. Kurths is a fellow of the American Physical Society, the Royal Society of 1023
Edinburgh, and the Network Science Society. He is a member of the Academia 1024
Europaea. He was the recipient of the Alexander von Humboldt
Research Award from India, in 2005, and from Poland in 2021, the Richardson medal of the European
Geophysical Union in 2013, and eight Honorary Doctorates. He is a highly
cited Researcher in Engineering. He is editor-in-chief of CHAOS and on
the editorial boards of more than 10 journals.
\end{IEEEbiography}

\end{document}